\documentclass[lettersize,journal]{IEEEtran}
\usepackage{hyperref}
\usepackage{amsmath,amsfonts}
\usepackage{algorithmic}
\usepackage{algorithm}
\usepackage{array}
\usepackage{textcomp}
\usepackage{stfloats}
\usepackage{url}
\usepackage{graphicx}
\usepackage{cite}
\hypersetup{hidelinks}

\usepackage{mathtools} 
\usepackage{amsthm} 
\usepackage{xcolor} 
\usepackage{enumitem} 

\allowdisplaybreaks[4] 

\input{preamble.sty}

\newcommand{\algorithmicparameters}{\textbf{Parameters:}}
\newcommand{\PARAMS}{\item[\algorithmicparameters]}

\theoremstyle{plain} 
\newtheorem{theorem}{Theorem}
\newtheorem{lemma}[theorem]{Lemma}

\newtheorem{corollary}[theorem]{Corollary}

\theoremstyle{definition} 
\newtheorem{problem}{Problem}

\newtheorem{assumption}{Assumption}

\theoremstyle{remark} 

\newif\ifshowchanges

\showchangesfalse   

\usepackage[dvipsnames]{xcolor}
\usepackage{amsmath,amssymb}
\usepackage{environ}

\newcommand{\addsymbol}{
  \raisebox{0.05ex}{\ensuremath{\oplus}}}

\newcommand{\delsymbol}{
  \raisebox{0.05ex}{\ensuremath{\ominus}}}

\newcommand{\chgspace}{\,}

\ifshowchanges

\newcommand{\added}[1]{\textcolor{ForestGreen}{\addsymbol\chgspace #1}}

\newcommand{\deleted}[1]{\textcolor{BrickRed}{\delsymbol\chgspace #1}%
}
\newcommand{\deletedparagraph}[1]{\paragraph*{\textcolor{BrickRed}{\delsymbol\chgspace #1}}}

\else

\newcommand{\added}[1]{#1}
\newcommand{\deleted}[1]{}
\newcommand{\deletedparagraph}[1]{}

\fi

\ifshowchanges

\NewEnviron{addedblock}{
\par
\begingroup
\color{ForestGreen}

\noindent
\addsymbol\chgspace \BODY

\par
\endgroup
}

\NewEnviron{deletedblock}{
\par
\begingroup
\color{BrickRed}

\noindent
\delsymbol\chgspace

\BODY

\par
\endgroup
}

\else

\NewEnviron{addedblock}{
\BODY
}

\NewEnviron{deletedblock}{
}

\fi

\usepackage{subfigure}
\usepackage{makecell}

\begin{document}

\title{Bayesian Signal Component Decomposition via Diffusion-within-Gibbs Sampling}

\author{Yi Zhang,~\IEEEmembership{Member,~IEEE,} Rui Guo,~\IEEEmembership{Member,~IEEE,} Yonina C. Eldar,~\IEEEmembership{Fellow,~IEEE}
	\thanks{This research was supported by the European Research Council (ERC) under the European Union’s Horizon 2020 research and innovation program (grant No. 101000967) and by the Israel Science Foundation (grant No. 536/22). This work is also supported by Manya Igel Centre for Biomedical Engineering and Signal Processing. The authors are with the Faculty of Mathematics and Computer Science, Weizmann Institute of Science, Rehovot 7610001, Israel (Corresponding author: Rui Guo). A preliminary short version of this work has been \deleted{accepted by}\added{presented at} ICASSP 2026\added{~\cite{zhang_bayesian_2026}}.}
	\thanks{}}

\markboth{Journal of \LaTeX\ Class Files,~Vol.~14, No.~8, August~2021}%
{Shell \MakeLowercase{\textit{et al.}}: A Sample Article Using IEEEtran.cls for IEEE Journals}


\maketitle

\begin{abstract}
	In signal processing, the data collected from sensing devices is often a noisy linear superposition of multiple components, and the estimation of components of interest constitutes a crucial pre-processing step. In this work, we develop a Bayesian framework for signal component decomposition, which combines Gibbs sampling with plug-and-play (PnP) diffusion priors to draw component samples from the posterior distribution. Unlike many existing methods, our framework supports incorporating \added{component-wise }model-driven and data-driven \deleted{prior knowledge into the diffusion prior}\added{priors into diffusion models} in a unified manner. Moreover, the proposed posterior sampler allows component priors to be learned separately and flexibly combined \deleted{without retraining}\added{for different decomposition tasks at inference time}. Under suitable assumptions, the proposed Diffusion-within-Gibbs (DiG) sampler provably produces samples from the posterior distribution. We also show that DiG can be interpreted as an extension of a class of recently proposed diffusion-based samplers, and that, for suitable classes of sensing operators, DiG better exploits the structure of the measurement model. Numerical experiments demonstrate the superior performance of our method over existing approaches.
\end{abstract}


\section{Introduction}
\label{sec:introduction}

\IEEEPARstart{S}{ignal} processing emphasizes extracting useful information from measurements of physical quantities (e.g., mechanical and electromagnetic waves). Due to wave superposition and other complex physical phenomena, the data directly collected from sensing devices is often a superposition of multiple components, which may include
\begin{itemize}
	\item several signal components that each carry an individual, self-contained piece of information (for example, a microphone recording may contain overlapping speech from multiple speakers talking simultaneously \cite{luo2018}),

	\item undesired interference components (such as co-channel interference in communications \cite{stuber_principles_2017}, or clutter arising from background scatterers in ultrasound imaging \cite{solomon_deep_2020}),

	\item observation noise induced by the sensing hardware and other small random perturbations.
\end{itemize}

\noindent
To extract useful information from such mixtures, a typical first step in signal processing is to estimate the individual components from the measurements, and to then select the components that are relevant to the downstream task.

In many applications, the relation between the observed signal $\yv$ and the latent components can be accurately or approximately described by a linear superposition model:
\begin{equation}\label{eq:observation}
	\yv = \sum_{k=1}^K \Hm_k \sv_k + \vv,
\end{equation}
where $\sv_k$ denotes a single signal component of interest or an interference component, $\Hm_k$ is a linear operator describing the sensing process, and $\vv$ is the observation noise. The resulting component decomposition problem in \eqref{eq:observation}, i.e., estimation of $(\sv_1,\cdots,\sv_K)$ from $\yv$, is a linear inverse problem. However, due to its highly underdetermined nature, achieving high–quality estimation relies on exploiting structural properties of the components, which leads to two central challenges:
\begin{enumerate}[label=\arabic*)]
	\item how to obtain structural information about real-world signal components and accurately encode such information in a mathematical model, and

	\item how to efficiently incorporate the resulting mathematical model into the computational procedure of signal component decomposition.
\end{enumerate}
Over the past several decades, numerous approaches have been proposed to address these two challenges.

\subsection{Related Work}

In classical model-driven signal processing, prior knowledge about the signal components is primarily obtained from the engineer's understanding of the signal generation and sensing mechanisms. Such physics-inspired priors are carefully encoded as analytically specified constraints, regularization terms, or probability distributions, and are then embedded into tailored signal transform, optimization algorithms, or sampling algorithms for component decomposition.

\textit{Transform-based methods}, such as filter banks \cite{bamberger1992}, wavelet decompositions \cite{mallat1989}, and empirical mode decomposition (EMD \cite{huang1998}), design linear or nonlinear operators that directly decompose an observation into components with prescribed narrowband or approximately periodic structure. These approaches are highly effective when the underlying components exhibit simple, well-separated time–frequency characteristics, but it is difficult to extend such methods to more general signals or to scenarios involving complicated sensing operators.

To obtain more flexible component decomposition frameworks, \textit{optimization-based methods} have been developed. These approaches encode prior information into regularization terms and recover the components by solving optimization problems of the form:
\begin{equation}\label{eq:opt-decomp}
	\min_{\sv_1,\cdots,\sv_K} F
	\dparen{\yv}{\sum_{k=1}^K \Hm_k \sv_k } + \lambda R(\sv_1,\cdots,\sv_K),
\end{equation}
where $F\dparen{\cdot}{\cdot}$ is a data-fidelity term, $R$ is a regularizer encoding the prior structure of the components, and $\lambda>0$ is a tuning parameter. By appropriately designing the regularizer, the formulation enables a flexible combination of sparsity, smoothness, low-rank structure, and other properties within a unified variational formulation, which has led to influential methods such as variational mode decomposition \cite{dragomiretskiy2014}, morphological component analysis \cite{starck2005}, and sparse-plus-low-rank decompositions \cite{candes2011}. When the overall objective in \eqref{eq:opt-decomp} is convex, convergence guarantees to a global minimizer can often be established. However, to capture more realistic signal structures one typically resorts to highly nonconvex regularizers, in which case numerical algorithms may get trapped in local minima.

\textit{Sampling-based methods} adopt a fully Bayesian viewpoint and directly model the prior distribution of the components, aiming to draw approximate samples from the posterior
\begin{equation}
	p_{\sv_1,\cdots,\sv_K|\yv} \propto p_{\yv|\sv_1,\cdots,\sv_K} \times  p_{\sv_1,\cdots, \sv_K},
\end{equation}
using Markov chain Monte Carlo (MCMC \cite{brooks_handbook_2011}) techniques. Here $p_{\sv_1,\cdots,\sv_K}$ is a hand-crafted prior capturing structural properties of the components, and the likelihood $p_{\yv|\sv_1,\dots,\sv_K}$ is determined by the sensing model \eqref{eq:observation}. Since the regularizer $R$ in \eqref{eq:opt-decomp} often admits a probabilistic interpretation as a negative log-prior, many optimization-based methods naturally extend to a sampling framework \cite{horst_bayesian_2025, ding_bayesian_2011}. Moreover, probabilistic modeling allows hierarchical priors that are difficult to formulate deterministically, and standard MCMC convergence theory does not require convexity of the negative log posterior \cite{tierney1994}. Nonetheless, manually specifying the prior distribution often introduces many hyperparameters, and the quality of the decomposition is often highly sensitive to these choices.

While these model-driven methods provide clear structure assumptions and, in some cases, theoretical guarantees, their reliance on simple analytic priors limits their ability to faithfully capture the rich, heterogeneous statistics of real-world signal components, especially in highly underdetermined settings. In the past decade, the advent of deep learning has fundamentally changed how prior information is obtained in component decomposition, enabling learning priors directly from examples.

In transform-based approaches, hand-crafted decomposition transforms have been complemented or replaced by discriminative neural networks that learn a direct mapping from observations to components \cite{wang_multilevel_2018,yu_deep_2024,weninger2014}. While effective, such end-to-end models are typically problem-specific, and require retraining when the number \added{of components, the }\deleted{or }distribution of\added{ even a small subset of }\deleted{ the} components\added{,} or the sensing model changes. In optimization-based approaches, plug-and-play methods \cite{venkatakrishnan2013, romano_little_2017} replace the proximal operators of explicit regularizers in certain proximal-type algorithms (e.g., ADMM \cite{boyd_distributed_2010}) by powerful learned denoisers, while preserving modular treatment of the sensing operators; yet existing theory usually characterizes convergence only to a fixed point of the algorithm \cite{ryu_plug-and-play_2019}, and the precise relation of this fixed point to a well-defined optimization objective is often unclear.

More recently, generative models \added{have been widely used for inference and reconstruction tasks~\cite{liu2024detracking, zhang_irgen_2025}, which has}\deleted{have} opened the door to sampling-based component decomposition with learned priors. Diffusion models, as the state-of-the-art in generative models, allow posterior sampling in a plug-and-play manner for general inverse problems \cite{daras2024}. Some recently proposed diffusion posterior samplers \cite{coeurdoux_pnp_2024, xu2024,wu_principled_2024,dou2024} have been proven to be asymptotically consistent\footnote{A posterior sampler is said to be consistent if the distribution of the generated samples converge to the true posterior distribution.} when the diffusion models are perfectly trained (i.e., the learned priors match the true component priors). However, these methods are typically formulated for a single aggregate unknown and make no explicit use of the multi-component structure inherent in decomposition problems. As a result, when directly applied to component decomposition, they do not fully exploit the structure of the linear mixture, which tends to substantially limit their effectiveness in highly underdetermined settings (cf. Section~\ref{sec:experiments}).

To date, model-based and data-driven priors have typically been employed separately in component decomposition.

\subsection{Proposed Method}\label{subsec:proposed-method}

In this work, we develop a Bayesian framework for signal component decomposition that combines Gibbs sampling with plug-and-play (PnP) diffusion priors. The proposed framework has the following notable features and contributions:
\begin{enumerate}
	\item \deleted{We use diffusion models to encode the component priors in a unified manner, and show how model-driven prior information can be combined with data-driven learning within the diffusion training pipeline. This hybrid modeling enables the use of analytic priors when data are scarce, while still benefiting from learned statistics when representative samples are available.}\added{We use diffusion models to encode component priors in a unified manner, enabling component-wise combinations of learned and model-based priors within the same diffusion-based implementation framework. This hybrid modeling enables the use of analytic priors for components whose training data are scarce, while still benefiting from learned statistics for components with representative samples.}

	\item We introduce a modular posterior sampling algorithm for component decomposition, termed the diffusion-within-Gibbs (DiG) sampler. The DiG sampler alternates Gibbs updates of the individual components using their respective diffusion models, allowing each component prior to be trained independently and combined flexibly\added{ for different decomposition tasks} at inference time\deleted{ without retraining}.

	\item Under the assumption of perfectly trained diffusion models, we establish the asymptotic consistency of the DiG sampler. We further show that DiG can be interpreted as an extension of a class of recently proposed diffusion-based samplers, and that, for suitable classes of sensing operators, DiG better exploits the structure of the measurement model.

	\item We validate the proposed framework through numerical experiments. In particular, on the task of extracting heartbeat signals from motion-induced interference, we demonstrate that, when combined with appropriate model-based priors, the DiG sampler achieves superior decomposition quality while requiring substantially less training data than competing diffusion-based sampling methods designed for generic inverse problems.
\end{enumerate}

The remainder of this paper is organized as follows. Section~\ref{sec:preliminary} reviews the basic principles and implementation details of diffusion models. In Section~\ref{sec:prior-modeling}, we formalize the component decomposition problem, and discuss how to incorporate\deleted{ (both} model-driven and data-driven\deleted{)} component priors into diffusion models. Section~\ref{sec:DiG-sampling} develops the proposed diffusion-within-Gibbs (DiG) sampling algorithm for component decomposition. In Section~\ref{sec:theory}, we establish the asymptotic consistency of the DiG sampler,\added{ discuss the practical implications of the assumptions underlying the consistency theorem, and} clarify \deleted{its relationship to}\added{the relationship between DiG and} some previously proposed diffusion-based samplers\deleted{, and present useful tips on certain implementation issues}. Numerical results on synthetic and real-world examples are reported in Section~\ref{sec:experiments}, followed by concluding remarks in Section~\ref{sec:conclusion}.

\textbf{Notation:}
\added{We follow the common signal-processing convention that bold lowercase letters denote vectors and bold uppercase letters denote matrices.}
We use $\R$ and $\R_{++}$ to denote the sets of real numbers and strictly positive real numbers, respectively.
For a matrix $\Am\in\R^{m\times n}$, $\Am^\top$ denotes its transpose.
Given an ordered tuple $\paren{\sv_k}_{k=1}^K$, for $i\leq j$ we define
$\sv_{i:j}\coloneq \paren{\sv_i,\sv_{i+1},\ldots,\sv_j}$.
For an index set $\Hcal\subset\set{1,2,\ldots,K}$, we define
$\sv_{\Hcal}\coloneq \paren{\sv_k}_{k\in\Hcal}$ and
$\sv_{\neg\Hcal}\coloneq \paren{\sv_k}_{k\notin\Hcal}$.
In particular, for $1\leq k_0\leq K$ we write $\sv_{\neg k_0}\coloneq \sv_{\neg\set{k_0}}$.
For a random variable\added{, e.g.,} $\sv$, we use $\breve{\sv}$ to denote a realization of $\sv$\added{, and we apply the same breve convention whenever distinguishing a random vector from its realized value is necessary}.
We use $\NormalPdf{\cdot}{\muv}{\Sigmam}$ to denote the probability density function of a Gaussian distribution with mean $\muv$ and covariance $\Sigmam$.
Unless stated otherwise, $\nv$ denotes a standard Gaussian random vector that is independent of all other random variables under consideration.

\section{Preliminaries on Diffusion Models}
\label{sec:preliminary}

Stochastic differential equation (SDE)-based diffusion models \cite{song2021} have emerged as a powerful tool for modeling complex probability distributions underlying a collection of data samples. Let $p_{\text{data}}$ denote the data distribution. By gradually injecting noise through a forward diffusion process, one can transform samples from $p_{\text{data}}$ into samples whose distribution approaches a simple Gaussian noise distribution. A diffusion model then generates new samples by approximately simulating the corresponding reverse-time diffusion dynamics, which map Gaussian noise back to the data distribution.

\subsection{Forward and Reverse-Time Diffusion Processes}
\label{sec:forward_reverse_SDEs}

The forward diffusion process is a Markov process $\paren{\xv_t}_{t=0}^T$ with initial state $\xv_0 \sim p_{\text{data}}$ and transition kernel given by the following SDE\footnote{In the general case \cite{song2021}, the forward process is written as $d\xv_t= f(\xv_t,t)dt + g(t)d\wv_t$. Here we set $f(\xv_t,t)\equiv \zerov$ for simplicity, which already covers many state-of-the-art diffusion models \cite{karras2022}.} \cite{song2021,karras2022}:
\begin{equation}\label{eq:forward_SDE}
	d\xv_t = g(t) d\wv_t,
\end{equation}
where $\wv_t$ is a Wiener process and $g(t)>0$ controls the noise injection rate at time $t$. Integrating (\ref{eq:forward_SDE}) yields
\begin{equation}\label{eq:marginal_pdf}
	\xv_t = \xv_0 + \int_{0}^t g(s) d\wv_s \sim p_{ \xv_0 + \sigma(t) \nv} =: p_{\sigma(t)},
\end{equation}
where $\nv\sim\Normal{\zerov}{\Imat}$ is independent of $\xv_0$, and
\begin{equation}\label{eq:sigma_t}
	\sigma(t)\coloneq \sqrt{ \int_{0}^t g(s)^2ds }
\end{equation}
denotes the standard deviation of the injected noise. Notice that the marginal distribution of $\xv_t$ only depends on $\sigma(t)$. When $\sigma(T)$ is much larger than the standard deviation of $p_{\text{data}}$, we may approximate $p_{\xv_T}$ by a Gaussian distribution $\Normal{\zerov}{\sigma(T)^2\Imat}$. Thus as $\xv_t$ evolves from $t=0$ to $T$, the marginal distribution of $\xv_t$ transforms from $p_{\xv_0}=p_{\text{data}}$ to a Gaussian distribution $\Normal{\zerov}{\sigma(T)^2\Imat}$.

The forward process (\ref{eq:forward_SDE}) admits a reverse-time counterpart \cite{anderson1982}. The reverse process is a Markov process $\paren{\bar{\xv}_t}_{t=0}^T$ independent of $\paren{\xv_t}_{t=0}^T$ with initial state being $\bar{\xv}_T\deq\xv_T\sim p_{\sigma(T)}$ and transition kernel given by the reverse-time SDE
\begin{equation}\label{eq:reverse_SDE}
	d\bar{\xv}_t = { - g(t)^2 \nabla_{\bar{\xv}_t} \log p_{\sigma(t)}(\bar{\xv}_t) }dt + g(t) d\bar{\wv}_t,
\end{equation}
where $\bar{\wv}_t$ is a Wiener process independent of $\wv_t$. For any positive integer $N$ and any finite set of times $0\leq t_1< t_2 < \cdots < t_N \leq T$, it is known that \cite{anderson1982}
\begin{equation}\label{eq:Anderson}
	\paren{\bar{\xv}_{t_1},\bar{\xv}_{t_2},\cdots,\bar{\xv}_{t_N}}\deq \paren{{\xv}_{t_1},{\xv}_{t_2},\cdots,{\xv}_{t_N}},
\end{equation}
i.e., if one observes $\bar{\xv}_t$ at a sequence of time points, the joint statistics are indistinguishable from those of $\xv_t$. Thus running $\bar{\xv}_t$ backward in time faithfully reproduces the law of $\xv_t$, which explains why it is referred to as the time reversal of $\xv_t$.

In particular, as $\bar{\xv}_t$ evolves from $t=T$ to $0$, according to (\ref{eq:marginal_pdf}) and (\ref{eq:Anderson}), the marginal distribution of $\bar{\xv}_t$ transforms from $p_{\bar{\xv}_T}=p_{\xv_T} \approx \Normal{\zerov}{\sigma(T)^2\Imat}$ to $p_{\bar{\xv}_0}=p_{\xv_0}=p_{\text{data}}$. Since sampling from $p_{\bar{\xv}_T} \approx \Normal{\zerov}{\sigma(T)^2\Imat}$ is straightforward, simulating \eqref{eq:reverse_SDE} from $t=T$ to $0$ yields samples from $p_{\text{data}}$.

\subsection{Implementation of Diffusion Models}
\label{sec:implementation_dm}

In (\ref{eq:reverse_SDE}), $\nabla \log p_{\sigma(t)}(\cdot)$ (termed \textit{score function}) has no closed-form expression and needs to be approximated by a neural network. By Tweedie's formula \cite{efron2011}, for any $\eta>0$,
\begin{equation}\label{eq:Tweedie}
	\paren{\forall \zv \in \R^d} \;\nabla_{\zv} \log p_\eta(\zv) = \frac{ \CondExpt{ \xv_0 }{ \xv_0 + \eta\nv =\zv } - \zv } { \eta^2 },
\end{equation}
where $\nv\sim\Normal{\zerov}{\Imat}$. Notice that $\CondExpt{ \xv_0 }{ \xv_0 + \eta\nv =\zv }$ is the MMSE estimator of $\xv_0$ given the noisy observation $\xv_0+\eta\nv$, thus it can be approximated by a denoising neural network $\CondExpt{ \xv_0 }{ \xv_0 + \eta\nv =\zv }\approx D_\theta(\zv;\eta)$, leading to
\begin{equation}\label{eq:score_approx}
	\nabla_{\zv} \log p_\eta (\zv) \approx \frac{ D_{\theta}(\zv; \eta) - \zv } { \eta^2 }.
\end{equation}
The reverse-time diffusion process (\ref{eq:reverse_SDE}) is then simulated by numerical integration. A common choice is the Euler–Maruyama method, which discretizes $[0,T]$ into steps $T=t_0>t_1>\cdots>t_{M}=0$. Starting from $\breve{\xv}_{t_0}\sim\Normal{\zerov}{\sigma(T)^2\Imat}$, we can simulate (\ref{eq:reverse_SDE}) by updating
\begin{align}
	\breve{\xv}_{t_{i+1}} \gets \breve{\xv}_{t_i} + g(t_i)^2 \frac{ D_{\theta}(\breve{\xv}_{t_i}; \sigma(t_i)) - \breve{\xv}_{t_i} } { \sigma(t_i)^2 } h_i + g(t_i)\sqrt{h_i}\,\epsv_i,  \label{eq:diff-integral-solver}
\end{align}
where $h_i\coloneq t_i-t_{i+1}$, $\epsv_i\sim\Normal{\zerov}{\Imat}$, and the score function in (\ref{eq:reverse_SDE}) is replaced by (\ref{eq:score_approx}). Then $\breve{\xv}_{t_M}$ serves as an approximate sample from $p_{\text{data}}$. In practice, the denoising network and numerical solver together constitute a diffusion model.

\section{Prior Modeling}
\label{sec:prior-modeling}

In this section, we first formalize the signal component decomposition problem and state the basic assumptions used throughout the paper. We then discuss how to incorporate\deleted{ both} model-driven and data-driven component prior information into diffusion models. The proposed hybrid prior modeling mechanisms play a key practical role in enabling high-quality decompositions in regimes where \deleted{only limited training data are available}\added{training data are available for only a subset of the components}.

\subsection{Problem Formulation and Assumptions}
\label{subsec:formulation}

As described in Section~\ref{sec:introduction}, we consider a linear sensing model in which the observed signal
\(\yv \in \mathbb{R}^m\) is a noisy linear mixture of \(K\) latent components \(\sv_1,\sv_2,\dots,\sv_K\):
\begin{equation}
	\yv \;=\; \sum_{k=1}^{K} \Hm_k \sv_k + \vv. \tag{\ref{eq:observation}}
\end{equation}
Here we adopt a Bayesian viewpoint, i.e., we regard each \(\sv_k \in \mathbb{R}^{d_k}\) as a random vector, with \(\Hm_k \in \mathbb{R}^{m \times d_k}\) being a known sensing matrix and \(\vv \in \mathbb{R}^m\) being random observation noise. We further make the following assumptions.

\begin{assumption}[Available prior knowledge]\label{assump:prior}
	For each \(1 \leq k \leq K\), we are given
	\begin{itemize}
		\item either an analytic form\footnote{The analytic form $f_k$ does not need to be a proper probability density function; in particular, we do not require $f_k$ to be normalized, or even to be integrable over $\R^{d_k}$.} \(f_k(\cdot)\) that approximates the prior density \(p_{\sv_k}\) of the \(k\)-th component up to a constant, or

		\item a set of i.i.d.\ samples
		      \(\{\breve{\sv}_{k,l}\}_{l=1}^{L_k}\) drawn from \(p_{\sv_k}\), from which \(p_{\sv_k}(\cdot)\) can be estimated in a data-driven manner.
	\end{itemize}
\end{assumption}

\begin{assumption}[Noise statistics]\label{assump:noise}
	The observation noise $\vv$ follows a Gaussian distribution \(\vv \sim \mathcal{N}(\mathbf{0}, \sigma_v^2 \Imat)\), and the variance \(\sigma_v^2\) is assumed to be known.
\end{assumption}

\begin{assumption}[Component independence]\label{assump:independence}
	The random vectors \(\sv_1,\sv_2,\dots,\sv_K\) and \(\vv\) are mutually independent.
\end{assumption}

Our objective is formulated as follows.

\begin{problem}[posterior sampling for component decomposition]\label{prob:post_samp}
Given a realization \(\breve{\yv}\) of the observation \(\yv\),
draw samples from the joint posterior distribution
$p_{\sv_1,\dots,\sv_K \mid \yv}(\cdot \mid \breve{\yv})$.
\end{problem}

Note that we explicitly target posterior sampling rather than directly computing point estimates of the components. In particular, posterior samples can be used to approximate the MMSE estimators of \(\sv_1,\dots,\sv_K\) (via sample averages), as well as more sophisticated uncertainty quantification metrics such as credible intervals \cite{edwards1963bayesian}, thereby yielding point estimates together with associated confidence measures.

In the following, we highlight three points to illustrate the generality of the above problem formulation and to justify the modeling assumptions.

\paragraph*{1) Correlated components} In some decomposition problems, several latent components may be strongly correlated. Suppose that, in \eqref{eq:observation}, a group of components \(\sv_{k_1},\sv_{k_2},\dots,\sv_{k_p}\) are statistically dependent within the group, but are jointly independent of the remaining components\added{ $\{\sv_k\}_{k\notin\{k_1,\dots,k_p\}}$} and the noise\added{ $\vv$}. Then one possible approach is to define an effective component $\sv'$ and its corresponding sensing matrix $\Hm'$ as
\begin{align*}
	\sv' & \coloneqq
	\begin{bmatrix}
		\sv_{k_1}^\top & \sv_{k_2}^\top & \cdots & \sv_{k_p}^\top
	\end{bmatrix}^\top, \\[1mm]
	\Hm' & \coloneqq
	\begin{bmatrix}
		\Hm_{k_1} & \Hm_{k_2} & \cdots & \Hm_{k_p}
	\end{bmatrix},
\end{align*}
and treat $\sv'$ as a new component to be inferred, whereby the observation model is rewritten as
\[
	\yv
	= \sum_{k \notin \{k_1,\dots,k_p\}} \Hm_k \sv_k
	+ \Hm' \sv' + \vv
\]
with \deleted{all the components involved}\added{$\set{\sv_k}_{k\not\in\set{k_1,\dots,k_p}}$, $\sv'$, and $\vv$} being\added{ mutually} independent. Hence the component independence assumption still applies at the level of component groups.

\paragraph*{2) Blind or uncertain sensing operators} In some scenarios, the sensing matrices \(\Hm_k\) may be unknown or slowly varying (e.g., in multichannel audio recording, \(\Hm_k\) may depend on the source positions). If the sensing matrix $\Hm_{k_0}$ in \eqref{eq:observation} is unknown, we suggest introducing an effective component
\[
	\sv' \coloneqq \Hm_{k_0}\sv_{k_0},
\]
and treating \(\sv'\) as the new latent variable to be inferred. By absorbing the uncertainty in \(\Hm_{k_0}\) into the prior of \(\sv'\), our formulation applies to a class of blind or partially blind settings where the goal is to decompose the mixture into additive contributions, rather than to identify the sensing matrices themselves.

\begin{deletedblock}
	\paragraph*{3) Non-Gaussian observation noise} In practice, the measurement noise may be non-Gaussian (e.g., Poisson noise in tomographic imaging \cite{bouman_unified_1996}). Such non-Gaussian noise can often be modeled explicitly as an additional independent component \(\sv_k\) in \eqref{eq:observation}. In this view, the Gaussian term \(\vv \sim \mathcal{N}(0,\sigma_v^2\Imat)\) in \eqref{eq:observation} is primarily a residual error term that accounts for model mismatch and small perturbations, thus the precise choice of \(\sigma_v^2\) is not critical and can be fixed for convenience without limiting the generality of the formulation.
\end{deletedblock}

\begin{addedblock}
	\paragraph*{3) Non-Gaussian, structured, or correlated observation noise} In practice, the observation noise may be non-Gaussian (e.g., Poisson noise in tomographic imaging \cite{bouman_unified_1996}), structured (e.g., noise with block structure), or correlated with signal components. One possible approach is to model such noise as an additional component $\sv_k$ in \eqref{eq:observation}, equipped with an appropriate prior model or dataset, while retaining $\vv\sim\Normal{\zerov}{\sigma_v^2\Imat}$ as an independent Gaussian residual perturbation. Once the observation noise is treated as a component, its distributional and structural properties can be incorporated into the design or training of the corresponding diffusion-model denoiser (cf. Section~\ref{subsec:train_diff_model}). If the noise is correlated with some signal components, it can be grouped with those components into a joint effective component, following the correlated-component construction above. Under this interpretation, non-Gaussian, structured, or correlated noise can be handled within the same observation model \eqref{eq:observation} through problem reformulation, rather than by modifying the subsequent sampling algorithm itself.
\end{addedblock}

\subsection{Incorporating Component Priors into a Diffusion Model}
\label{subsec:train_diff_model}

We next discuss how to incorporate the priors of individual signal components into diffusion models. As reviewed in Section~\ref{sec:implementation_dm}, for each \(1 \leq k \leq K\), implementing a diffusion model that samples from \(p_{\sv_k}\) essentially reduces to obtaining a denoiser \(D^{(k)}_\theta(\cdot;\eta)\) that approximates the MMSE estimator at all noise levels \(\eta>0\), i.e.,
\begin{align}
	D^{(k)}_\theta(\sv_k + \eta \nv;\eta) \approx \CondExpt{\sv_k}{\sv_k + \eta \nv} \nonumber \\
	= \underset{D}{ \argmin }~  \Expt{ \norm{ \sv_k - D(\sv_k + \eta\nv) }^2_2 }, \label{eq:MMSE-def}
\end{align}
where \(\nv\sim\Normal{\zerov}{\Imat}\) is independent of \(\sv_k\), and the minimization in \eqref{eq:MMSE-def} is taken over all measurable functions $D$.

In the machine learning literature, the denoiser $D^{(k)}_\theta\paren{\cdot;\eta}$ is typically learned from a set of i.i.d.\ samples \(\{\breve{\sv}_{k,l}\}_{l=1}^{L_k}\) drawn from \(p_{\sv_k}\). The population risk in \eqref{eq:MMSE-def} is then approximated by the empirical loss
\begin{equation}
	\frac{1}{L_k} \sum_{\ell=1}^{L_k}
	\mathbb{E}_{\nv \sim \mathcal{N}(\mathbf{0},\Imat)}
	\Bigl[
		\bigl\|\breve{\sv}_{k,\ell}
		- D^{(k)}_\theta\bigl(\breve{\sv}_{k,\ell} + \eta \nv;\eta\bigr)
		\bigr\|_2^2
		\Bigr],
	\label{eq:empirical-SE}
\end{equation}
and the network parameters \(\theta\) are optimized to minimize \eqref{eq:empirical-SE} over a range of noise levels \(\eta\).

However, such a purely data-driven training strategy is not fully aligned with the requirements of many signal processing applications. In numerous domains such as medical imaging\added{ \cite{solomon_deep_2020}}, it is often difficult to obtain large numbers of clean component samples for each \(\sv_k\). Moreover, classical signal processing research has provided rich (though simplified) descriptions of component structure, and ignoring such model-based prior information is obviously wasteful. \deleted{Below we discuss two relevant scenarios, and present practical strategies that take models into account.}\added{This motivates a hybrid prior modeling strategy for component priors: data-driven diffusion models can be used when representative component samples are available, whereas model-based priors can be used when training data are scarce but reliable structural knowledge is available.}

\begin{deletedblock}
	\paragraph*{1) Data-plus-model scenarios}
	When in addition to a collection of training samples $\set{\breve{\sv}_{k,l}}_{l=1}^{L_k}$, an analytic model-based prior \(f_k\) is also available, a natural way to incorporate $f_k$ into the training of \(D^{(k)}_\theta\) is to penalize outputs that lie in low-density regions under \(f_k\). Concretely, assume that $f_k$ is differentiable almost everywhere and that its gradient can be computed efficiently, we propose training the denoiser by minimizing the regularized empirical risk
	\begin{align}
		L(\theta) \coloneq & \inv{L_k} \sum_{l=1}^{L_k} \Ebb_{\nv\sim\Normal{\zerov}{\Imat}} \bigg[ \norm{ \breve{\sv}_{k,l} - D^{(k)}_\theta\paren{\breve{\sv}_{k,l} + \eta \nv ; \eta} }^2_2  \nonumber \\
		                   & \hspace{7em} - \lambda \log f_k \paren{ D^{(k)}_\theta\paren{\breve{\sv}_{k,l} + \eta \nv ; \eta } } \bigg], \label{eq:data-plus-model-DM-training}
	\end{align}
	where \(\lambda>0\) is a tuning parameter.
\end{deletedblock}

\paragraph*{\deleted{2) Model-only scenarios}\added{Model-based component priors}}
\deleted{Sometimes it may be impossible to obtain isolated observations of individual components (e.g., clutter suppression in ultrasound imaging \cite{solomon_deep_2020}). In this case, a straightforward solution is to train \(D^{(k)}_\theta\) on synthetic data generated from a sufficiently accurate prior model. However, this requires a strong and realistic model for \(\sv_k\). When only relatively simple model knowledge is available,}\added{Model-based prior information can enter the diffusion framework through two routes. The first route is simulation-based: when a sufficiently realistic prior model is available, one may generate a large collection of synthetic component samples from this model and train \(D^{(k)}_\theta\) using the empirical denoising loss in \eqref{eq:empirical-SE}. In this route, model-based knowledge enters the diffusion model indirectly through the synthetic training distribution.}

\added{In the second route, the model knowledge is used directly in computing the denoising operator. When only relatively simple model knowledge is available, rather than learning \(D^{(k)}_\theta\) from potentially misspecified synthetic data,} a more practical alternative is to use a hand-crafted\added{ yet computationally efficient approximate} denoiser, e.g., the MAP estimator
\begin{equation}
	D^{(k)}_{\text{MAP}}(\zv;\eta)
	\coloneq
	\arg\min_{\breve{\sv}_k}
	\biggl\{
	\frac{1}{2\eta^2}\bigl\|\zv - \breve{\sv}_k\bigr\|_2^2
	- \log f_k(\breve{\sv}_k)
	\biggr\}, \label{eq:model-only-DM-training}
\end{equation}
in place of the parameterized denoiser \(D^{(k)}_\theta(\zv;\eta)\)\added{ at all noise levels required by the diffusion model}. Although this \added{MAP-based} approximation \added{is not equivalent to the true MMSE denoiser and} introduces\added{ additional} approximation error\deleted{ relative to the true MMSE denoiser}, for simple priors \(f_k\) (e.g., smoothness-promoting priors) the MAP estimator may admit a closed-form expression or a very efficient solver, allowing us to substitute a cheap \deleted{analytic}\added{model-based} denoiser for the denoising network and thereby\deleted{ substantially} reduce \deleted{computational cost}\added{both training and inference costs}.

\deleted{As illustrated in Section~\ref{subsec:heartbeat-extraction} for the task of extracting heartbeat signals from motion-induced interference, when assisted with proper model-based priors, the techniques above enable us to dramatically reduce the amount of training data and computational time, with only minor loss in decomposition performance.}\added{The second, direct route is validated in Section~\ref{subsec:heartbeat-extraction} for the task of extracting heartbeat signals from motion-induced interference. The results demonstrate that, when assisted with appropriate model-based priors, the hybrid prior modeling strategy can effectively reduce the dependence on training data and improve computational efficiency.}

After \deleted{training}\added{obtaining} diffusion models for sampling $p_{\sv_k}$, we entirely replace the prior information specified in Assumption~\ref{assump:prior} with diffusion-model-based priors, which will then be used to solve the subsequent posterior sampling problem (Problem~\ref{prob:post_samp}). To facilitate the derivation of the proposed posterior sampling algorithm in the next section, we introduce an idealized assumption of perfect diffusion models, meaning that the reverse diffusion processes are simulated exactly. This assumption allows us to isolate the algorithmic structure from errors due to imperfect score estimation and numerical integration.

\begin{assumption}[Perfect diffusion models]\label{assump:diff}
	For any $1\leq k\leq K$, consider the following reverse-time SDE $\paren{ \bar{\xv}_{k,t} }_{t=0}^T$ for sampling $p_{\sv_k}$:
	\begin{equation}\label{eq:reverse-SDE-sk}
		d\bar{\xv}_{k,t} = { - g(t)^2 \nabla_{\bar{\xv}_{k,t}} \log p^{(k)}_{\sigma(t)}(\bar{\xv}_{k,t}) }dt + g(t) d\bar{\wv}_{k,t},
	\end{equation}
	where $g(t)\added{>0}$ is a user-specified noise schedule function,\footnote{In fact, one may choose different terminal time $T$ and noise schedules $g(t)$ for different $k$ (accordingly, $\sigma(t)$ will also depend on $k$ via \eqref{eq:sigma_t}). We do not explicitly write such dependence on $k$ for simplicity.} $\sigma(\cdot)\coloneq \sqrt{\int_{0}^{\paren{\cdot}} g(s)^2 ds }$ is the cumulative injected-noise standard deviation\added{ satisfying $\sigma(T)\gg \sigma_v$}, and for any $t\in[0,T]$,
	\begin{equation}\label{eq:marginal-diff-sk}
		p^{(k)}_{\sigma(t)} \coloneq p_{\bar{\xv}_{k,t}} \coloneq p_{\sv_k + \sigma(t) \nv}
	\end{equation}
	denotes the marginal distribution of $\bar{\xv}_{k,t}$. We assume that there exists an idealized diffusion model such that, for any pair of time instants $0\leq t_1 \leq t_2 \leq T$, the trajectory of $\bar{\xv}_{k,t}$ from $t_2$ to $t_1$ can be simulated exactly, i.e., for any realization $\breve{\xv}_{k,t_2}$ of the initial state $\bar{\xv}_{k,t_2}$, the reverse-time transition $p_{\bar{\xv}_{k,t_1}| \bar{\xv}_{k,t_2} }\cparen{\cdot}{ \breve{\xv}_{k,t_2} }$ is exactly sampleable.
\end{assumption}

\section{Diffusion-within-Gibbs Sampling}\label{sec:DiG-sampling}

\begin{algorithm}[t]
	\caption{Gibbs sampling for Problem~\ref{prob:post_samp}}\label{alg:Gibbs}
	\begin{algorithmic}
		\REQUIRE {Initial sample} $\sv^{(0)}_{1:K}$, {observation} $\breve{\yv}$
		\PARAMS {Number of iterations} $N$
		\ENSURE {Approximate posterior sample} $\breve{\sv}_{1:K}$
		\STATE $\breve{\sv}_{1:K}\gets \sv^{(0)}_{1:K}$
		\FOR{$i=1,2,\cdots,N$}
		\FOR{$k=1,2,\cdots,K$}
		\STATE {Draw} $\breve{\sv}_k  \sim {p}_{\sv_{k} \mid \yv, \sv_{\neg k}}\cparen{\cdot}{\breve{\yv}, \breve{\sv}_{\neg k}}$
		\ENDFOR
		\ENDFOR
		\RETURN $\breve{\sv}_{1:K}$
	\end{algorithmic}
\end{algorithm}

In this section, we build upon the component-wise diffusion priors developed in the previous section, and combine them with a Gibbs sampler to derive a diffusion-within-Gibbs algorithm for solving the component decomposition problem in Problem~\ref{prob:post_samp}.

Gibbs sampling \cite{geman_stochastic_1984} is a standard approach for sampling a multivariate joint distribution such as $p_{\sv_{1:K}|\yv}$ (cf. Algorithm~\ref{alg:Gibbs}), which iteratively updates one component of $\sv_{1:K}$ at a time by sampling from its conditional distribution given the others. However, for general prior distributions $p_{\sv_k}$ ($1\leq k\leq K$), the required conditional samplers are intractable, posing a key challenge for practical implementation.

In this section, we show that for the observation model in \eqref{eq:observation}, these conditional sampling steps are achievable via partial simulation of the reverse diffusion process (\ref{eq:reverse-SDE-sk}).

\subsection{Conditional Sampling via Diffusion When $\Hm_k=\Imat$}

Fix an index \(k\). To draw samples from \(p_{\sv_k \mid \yv,\sv_{\neg k}}\), we first consider the simple case where \(\Hm_k=\Imat\), and show that the sampling of $p_{\sv_k|\yv,\sv_{\neg k}}$ is achievable via diffusion models. By Bayes' rule,
\begin{align}
	p_{\sv_k | \yv, \sv_{\neg k}} \cparen{ \cdot }{ \breve{\yv}, \breve{\sv}_{\neg k} } & \propto  p_{\sv_k | \sv_{\neg k}} \cparen{ \cdot }{ \breve{\sv}_{\neg k} } \times p_{\yv|\sv_k, \sv_{\neg k}} \cparen{\breve{\yv}}{ \cdot, \breve{\sv}_{\neg k} } \nonumber \\
	                                                                                    & \propto p_{\sv_k}(\cdot) \times  e^{-\inv{2\sigma_v^2}\norm{ \breve{\rv}_{\neg k} - \paren{\cdot} }^2_2}, \label{eq:cond_Gibbs_temp1}
\end{align}
where \(\breve{\rv}_{\neg k} \coloneq \breve{\yv} - \sum_{j\neq k} \Hm_j \breve{\sv}_j\) is the residual computed from \(\breve{\sv}_{\neg k}\).
Equation~\eqref{eq:cond_Gibbs_temp1} follows from the component-independence assumption (Assumption~\ref{assump:independence}), the observation model \eqref{eq:observation}, and the condition \(\Hm_k=\Imat\).

On the other hand, let \(\nv\sim \Normal{\zerov}{\Imat}\) be independent of all other random variables, and consider the following auxiliary conditional density:
\begin{align}
	p_{\sv_k|\sv_k+\sigma_v \nv}\cparen{ \cdot }{ \breve{\rv}_{\neg k} } & \propto p_{\sv_k}(\cdot) \times p_{\sv_k + \sigma_v \nv|\sv_k}\cparen{ \breve{\rv}_{\neg k} }{ \cdot } \nonumber                      \\
	                                                                     & \propto p_{\sv_k}(\cdot) \times  e^{-\inv{2\sigma_v^2}\norm{ \breve{\rv}_{\neg k} - \paren{\cdot} }^2_2}. \label{eq:cond_Gibbs_temp2}
\end{align}

Comparing \eqref{eq:cond_Gibbs_temp1} and \eqref{eq:cond_Gibbs_temp2}, we obtain
\begin{equation}\label{eq:cond-Gibbs-alter-form}
	p_{\sv_k \mid \yv, \sv_{\neg k}} \cparen{ \cdot }{ \breve{\yv}, \breve{\sv}_{\neg k} }
	= p_{\sv_k\mid \sv_k+\sigma_v \nv}\cparen{ \cdot }{ \breve{\rv}_{\neg k} },
\end{equation}
where \(\breve{\rv}_{\neg k} \coloneq \breve{\yv} - \sum_{j\neq k} \Hm_j \breve{\sv}_j\). \eqref{eq:cond-Gibbs-alter-form} reveals an important fact: {when $\Hm_k=\Imat$, the Gibbs conditional update is equivalent to sampling from a denoising posterior $p_{\sv_k\mid \sv_k+\eta\nv}$ at a specific noise level $\eta>0$}. Interestingly, it can be further shown (see Lemma~\ref{lemma:denoise-diff}) that {for any $\eta>0$, sampling from \(p_{\sv_k\mid \sv_k+\eta\nv}\) is equivalent to simulating the reverse-time diffusion process $\bar{\xv}_{k,t}$ in Assumption~\ref{assump:diff} from a specific starting time down to $t=0$}, hence leading to a diffusion-based implementation of the Gibbs update.

\begin{lemma}\label{lemma:denoise-diff}
	Let \(\paren{\bar{\xv}_{k,t}}_{t=0}^T\) be defined as in Assumption~\ref{assump:diff}. Then given \(\eta\deleted{>0}\added{\in(0,\sigma(T)]}\), for any \(\zv\in\R^{d_k}\), we have
	\begin{equation}
		p_{\sv_k \mid \sv_k + \eta \nv} \cparen{ \cdot }{ \zv }
		= p_{\bar{\xv}_{k,0} \mid \bar{\xv}_{k, \sigma^{-1}(\eta)}} \cparen{ \cdot }{ \zv },
	\end{equation}
	where \(\sigma^{-1}\) denotes the inverse of the noise-scale function \(\sigma(\cdot)\) in Assumption~\ref{assump:diff}.
\end{lemma}

\begin{IEEEproof}
	\added{From Assumption~\ref{assump:diff}, we have $\sigma(\cdot)\coloneq \sqrt{\int_{0}^{\paren{\cdot}} g(s)^2 ds }$ and $g(t)>0$, thus $\sigma(\cdot)$ is strictly increasing and hence invertible on $[0,T]$. }Let \(\paren{ \xv_{k,t} }_{t=0}^T\) denote the forward diffusion process associated with \(\paren{ \bar{\xv}_{k,t} }_{t=0}^T\). By \eqref{eq:Anderson} and \eqref{eq:marginal-diff-sk}, we have
	\begin{align*}
		p_{\bar{\xv}_{k,0} \mid \bar{\xv}_{k, t_\eta}} \cparen{ \cdot }{ \zv }
		 & = p_{\xv_{k,0}\mid\xv_{k,t_\eta}} \cparen{ \cdot }{ \zv } \\
		 & = p_{\sv_k \mid \sv_k + \sigma(t_\eta)\nv}
		\cparen{\cdot}{\zv}                                          \\
		 & = p_{\sv_k \mid \sv_k + \eta \nv} \cparen{\cdot}{\zv},
	\end{align*}
	where \(t_\eta \coloneq \sigma^{-1}(\eta)\). This proves the lemma.
\end{IEEEproof}

By \deleted{\eqref{eq:cond-Gibbs-alter-form} and Lemma~\ref{lemma:denoise-diff}}\added{\eqref{eq:cond-Gibbs-alter-form}, Lemma~\ref{lemma:denoise-diff} and $\sigma_v \ll \sigma(T)$ (Assumption~\ref{assump:diff})}, when \(\Hm_k=\Imat\), the conditional update in Gibbs sampling is implemented as follows: let \(t_v \coloneq \sigma^{-1}(\sigma_v)\), set \(\bar{\xv}_{k,t_v}\gets \breve{\rv}_{\neg k}\), and simulate the reverse-time SDE in \eqref{eq:reverse-SDE-sk} from \(t=t_v\) down to \(t=0\). The resulting realization of \(\bar{\xv}_{k,0}\) is then a valid sample from $p_{\sv_k|\yv,\sv_{\neg k}}$. When all \(\Hm_k\) are identity matrices (e.g., in single-channel source separation), this procedure yields a direct implementation of Algorithm~\ref{alg:Gibbs}.

When some sensing matrices are non-identity, a diffusion-model-based implementation of Algorithm~\ref{alg:Gibbs} is no longer straightforward. Nevertheless, we show next that the general case with \(\Hm_k\neq \Imat\) reduces to the identity case via a relaxation technique.

\subsection{Reducing General $\Hm_k$ to the Identity Case via Relaxation}\label{sec:relaxation-technique}

\begin{algorithm}[t]
	\caption{Gibbs sampling under \eqref{eq:observation-relax}}\label{alg:Gibbs-relax}
	\begin{algorithmic}
		\REQUIRE {Initial samples} $\sv^{(0)}_{1:K}$, $\uv^{(0)}_{\Hcal}$, {observation} $\breve{\yv}$
		\PARAMS {Number of iterations} $N$, {relaxation parameters} $\paren{\eta_k}_{k\in\Hcal}\subset\R_{++}$
		\ENSURE {Approximate posterior sample} $\breve{\sv}_{1:K}$
		\STATE $\breve{\sv}_{1:K} \gets \sv^{(0)}_{1:K}$
		\STATE $\breve{\uv}_{\Hcal} \gets \uv^{(0)}_{\Hcal}$
		\FOR{$i=1,2,\cdots,N$}
		\FOR{$k=1,2,\cdots,K$}
		\IF{$k\in\Hcal$}
		\STATE {Draw} $\breve{\uv}_k \sim {p}_{\uv_{k} \mid \yv', \sv_{1:K}, \uv_{\neg k}}
			\cparen{\cdot}{\breve{\yv}, \breve{\sv}_{1:K}, \breve{\uv}_{\neg k}}$
		\ENDIF
		\STATE {Draw} $\breve{\sv}_k \sim {p}_{\sv_{k} \mid \yv', \sv_{\neg k}, \uv_{\Hcal}}
			\cparen{\cdot}{\breve{\yv}, \breve{\sv}_{\neg k}, \breve{\uv}_{\Hcal}}$
		\ENDFOR
		\ENDFOR
		\RETURN $\breve{\sv}_{1:K}$
	\end{algorithmic}
\end{algorithm}

To handle the general case with \(\Hm_k\neq \Imat\), we consider a relaxed observation model that is slightly different from \eqref{eq:observation}:
\begin{equation}\label{eq:observation-relax}
	\yv' = \sum_{\substack{1\leq k\leq K \\ \Hm_k = \Imat}} \sv_k
	+ \sum_{\substack{1\leq k\leq K \\ \Hm_k \neq \Imat}} \Hm_k \paren{\sv_k + \vv_k }
	+ \vv,
\end{equation}
where \(\vv_k \sim \Normal{\zerov}{\eta_k^2 \Imat}\) is independent of all other random variables. It is easy to see that the only difference between \eqref{eq:observation-relax} and \eqref{eq:observation} is the additional perturbation \(\vv_k\) injected into components $\sv_k$ with \(\Hm_k\neq \Imat\), which introduces a modeling mismatch. In particular, when \(\eta_k=0\) for all \(k\), \eqref{eq:observation-relax} reduces to the original observation model \eqref{eq:observation}. \added{Notice that in \eqref{eq:observation-relax}, perturbation is introduced only when necessary, i.e., for components with non-identity sensing operators, since the identity case can already be handled directly without relaxation and introducing additional relaxation would only increase the mismatch to the original observation model.}

For the relaxed model \eqref{eq:observation-relax}, we introduce some auxiliary variables which ease the analysis. Define the set of component indices with non-identity sensing matrices as
\begin{equation}\label{eq:Hcal}
	\Hcal \coloneq \cset{ k\in\set{1,2,\cdots,K} }{ \Hm_k \neq \Imat },
\end{equation}
and for each \(k\in\Hcal\), define the \added{auxiliary} relaxation variable
\begin{equation*}
	\uv_k \coloneq \sv_k + \vv_k.
\end{equation*}
We then apply Gibbs sampling to draw samples from
\(p_{\sv_{1:K}, \uv_{\Hcal}\mid\yv'}\cparen{\cdot}{\breve{\yv}}\); see Algorithm~\ref{alg:Gibbs-relax}.
Since \(p_{\sv_{1:K}\mid\yv'}\cparen{\cdot}{\breve{\yv}}\) is a marginal of
\(p_{\sv_{1:K}, \uv_{\Hcal}\mid\yv'}\cparen{\cdot}{\breve{\yv}}\), discarding the \(\uv_{\Hcal}\) part of a joint sample from \(p_{\sv_{1:K}, \uv_{\Hcal}\mid\yv'}\cparen{\cdot}{\breve{\yv}}\) yields a sample from {the exact posterior \(p_{\sv_{1:K}\mid\yv'}\cparen{\cdot}{\breve{\yv}}\) of the perturbed observation model \eqref{eq:observation-relax}}. Notice that the modeling error of \eqref{eq:observation-relax} is totally determined by the hyperparameters $(\eta_k)_{k\in\Hcal}$, when all \(\eta_k^2\) are sufficiently small, a sample from \(p_{\sv_{1:K}\mid\yv'}\cparen{\cdot}{\breve{\yv}}\) serves as an accurate approximation to a sample from \(p_{\sv_{1:K}\mid\yv}\cparen{\cdot}{\breve{\yv}}\).

In the next two lemmas, we show that the two conditional updates in Algorithm~\ref{alg:Gibbs-relax} admit tractable implementations: updating \(\breve{\uv}_k\) reduces to sampling from a Gaussian distribution (Lemma~\ref{lemma:cpdf-uk}), and updating \(\breve{\sv}_k\) reduces to sampling from \(p_{\sv_k\mid \sv_k+\eta \nv}\) at an appropriate noise level $\eta>0$ (Lemma~\ref{lemma:cpdf-sk}), thus falling back to the identity case discussed in Lemma~\ref{lemma:denoise-diff}. \added{Therefore, by introducing a user-controlled perturbation (i.e., bias)
	into the observation model, we obtain a computationally tractable Gibbs
	update scheme that remains compatible with diffusion-based sampling even
	under non-identity sensing operators.}

\begin{lemma}\label{lemma:cpdf-uk}
	For any \(k\in\Hcal\), given \(\breve{\yv}\), \(\breve{\sv}_{1:K}\), and \(\breve{\uv}_{\neg k}\), the conditional update of \(\breve{\uv}_k\) in Algorithm~\ref{alg:Gibbs-relax} is equivalent to sampling from a Gaussian distribution:
	\begin{equation*}
		p_{\uv_k \mid \yv', \sv_{1:K}, \uv_{\neg k}}
		\cparen{\cdot}{\breve{\yv}, \breve{\sv}_{1:K}, \breve{\uv}_{\neg k}}
		= \NormalPdf{\cdot}{\muv_k}{\Sigmam_k},
	\end{equation*}
	where the mean \(\muv_k\) and covariance \(\Sigmam_k\) are given by
	\begin{align*}
		\Sigmam_k & \coloneq \paren{ \inv{\sigma_v^2} \Hm_k^\top \Hm_k + \inv{\eta_k^2} \Imat }^{-1},                              \\
		\muv_k    & \coloneq \Sigmam_k \paren{ \inv{\sigma_v^2} \Hm_k^\top \breve{\rv}'_{\neg k} + \inv{\eta_k^2} \breve{\sv}_k },
	\end{align*}
	and \(\breve{\rv}'_{\neg k}\coloneq \breve{\yv}
	- \sum_{j\not\in \Hcal} \breve{\sv}_j
	- \sum_{j\in\Hcal,\, j\neq k} \Hm_j \breve{\uv}_j\) denotes the residual in \eqref{eq:observation-relax} with the \(k\)-th term removed.
\end{lemma}

\begin{IEEEproof}
	See Appendix~\ref{proof:cpdf-uk}.
\end{IEEEproof}

\begin{lemma}\label{lemma:cpdf-sk}
	For \(1\leq k\leq K\), given \(\breve{\yv}\), \(\breve{\sv}_{\neg k}\), and \(\breve{\uv}_{\Hcal}\), the conditional update of \(\breve{\sv}_k\) in Algorithm~\ref{alg:Gibbs-relax} simplifies to
	\begin{align*}
		 & p_{\sv_k | \yv', \sv_{\neg k}, \uv_{\Hcal}}\cparen{\cdot}{ \breve{\yv}, \breve{\sv}_{\neg k}, \breve{\uv}_{\Hcal} } = \\
		 & \hspace{4em}
		\begin{cases}
			p_{\sv_k | \sv_k + \eta_k \nv}\cparen{\cdot}{\breve{\uv}_k},           & \text{if}~k\in\Hcal,     \\
			p_{\sv_k | \sv_k + \sigma_v \nv}\cparen{\cdot}{\breve{\rv}'_{\neg k}}, & \text{if}~k\not\in\Hcal,
		\end{cases}
	\end{align*}
	where for \(k\not\in\Hcal\), \(\breve{\rv}'_{\neg k}\coloneq \breve{\yv}
	- \sum_{j\not\in \Hcal,\, j\neq k} \breve{\sv}_j
	- \sum_{j\in\Hcal} \Hm_j \breve{\uv}_j\) is the residual in \eqref{eq:observation-relax} with the \(k\)-th term removed.
\end{lemma}

\begin{IEEEproof}
	See Appendix~\ref{proof:cpdf-sk}.
\end{IEEEproof}

Combining Lemma~\ref{lemma:cpdf-sk} with Lemma~\ref{lemma:denoise-diff}, we conclude that the conditional update of \(\breve{\sv}_k\) in Algorithm~\ref{alg:Gibbs-relax} is implementable via a warm-started simulation (with $\breve{\uv}_k$ for $k\in\Hcal$\added{ with $\eta_k\in (0, \sigma(T)]$,} or $\breve{\rv}'_{\neg k}$ for $k\not\in\Hcal$ being the initial state) of the reverse diffusion process $\bar{\xv}_{k,t}$ described in Assumption~\ref{assump:diff}. Plugging these results into Algorithm~\ref{alg:Gibbs-relax} yields the proposed diffusion-within-Gibbs sampler; see Algorithm~\ref{alg:DiG}. \deleted{In each sweep, Algorithm~\ref{alg:DiG} performs \(K\) conditional updates; each update runs a reverse diffusion with \(M\) steps and thus requires \(M\) evaluations of the corresponding denoising network.
	Hence, the per-sweep computational complexity is \(O(KM)\), while the memory complexity is \(O(K)\) for storing the \(K\) component variables, plus the storage of the denoising networks.}

\begin{algorithm}[t]
	\caption{Diffusion-within-Gibbs (DiG) sampler}\label{alg:DiG}
	\begin{minipage}{\linewidth}
		\begin{algorithmic}
			\REQUIRE {Initial samples} $\sv^{(0)}_{1:K}$, $\uv^{(0)}_{\Hcal}$, {observation} $\breve{\yv}$
			\PARAMS {Number of iterations} $N$, {relaxation parameters} $\paren{\eta_k}_{k\in\Hcal}\subset\deleted{\R_{++}}\added{(0,\sigma(T)]}$
			\ENSURE {Approximate posterior sample} $\breve{\sv}_{1:K}$
			\STATE $\breve{\sv}_{1:K} \gets \sv^{(0)}_{1:K}$
			\STATE $\breve{\uv}_{\Hcal} \gets \uv^{(0)}_{\Hcal}$
			\FOR{$i=1,2,\cdots,N$}
			\FOR{$k=1,2,\cdots,K$}
			\STATE $\breve{\rv}'_{\neg k}\gets \begin{cases}
					\breve{\yv} - \sum_{j\not\in \Hcal} \breve{\sv}_j - \sum_{\substack{j\in\Hcal \\ j\neq k}} \Hm_j \breve{\uv}_j, & \text{if}~k\in\Hcal,\\
					\breve{\yv} - \sum_{\substack{j\not\in \Hcal                                  \\j\neq k}} \breve{\sv}_j - \sum_{j\in\Hcal} \Hm_j \breve{\uv}_j, & \text{if}~k\not\in\Hcal.
				\end{cases}$
			\IF{$k\in\Hcal$}
			\STATE ${\Sigmam}_k \gets \paren{ \inv{\sigma_v^2} \Hm_k^\top \Hm_k + \inv{\eta_k^2} \Imat }^{-1}$
			\STATE ${\muv}_k \gets \Sigmam_k \paren{ \inv{\sigma_v^2} \Hm_k^\top \breve{\rv}'_{\neg k} + \inv{\eta_k^2} \breve{\sv}_k }$
			\STATE {Draw} $\breve{\uv}_k  \sim \NormalPdf{ \cdot }{ {\muv}_k }{ {\Sigmam}_k}$
			\ENDIF
			\STATE {Draw} $
				\breve{\sv}_k\sim
				\begin{cases}
					p_{\bar{\xv}_{k,0} \mid \bar{\xv}_{k,\sigma^{-1}(\eta_k)}}\cparen{\cdot}{\breve{\uv}_k},            & \text{if}~k\in\Hcal,     \\
					p_{\bar{\xv}_{k,0} \mid \bar{\xv}_{k, \sigma^{-1}(\sigma_v)}}\cparen{\cdot}{\breve{\rv}'_{\neg k}}, & \text{if}~k\not\in\Hcal.
				\end{cases}$
			\ENDFOR
			\ENDFOR
			\RETURN $\breve{\sv}_{1:K}$
		\end{algorithmic}
	\end{minipage}
\end{algorithm}

\subsection{\deleted{Parameter Annealing Technique}\added{Adaptive Warm-Up Strategy}}
\label{subsec:warm-up}

\deleted{In practice, Gibbs samplers are often equipped with a parameter-annealing strategy \cite{geman_stochastic_1984}. For the proposed DiG sampler, we suggest adopting the following annealing scheme: in the \(i\)-th iteration of Algorithm~\ref{alg:DiG}, we replace the observation-noise level \(\sigma_v\) with \(\sigma^{(i)}_v\geq \sigma_v\), and for each \(k\in\Hcal\), we replace the modeling mismatch level \(\eta_k\) with \(\eta^{(i)}_k\geq \eta_k\). As the iteration index \(i\) increases, we gradually decrease \(\sigma^{(i)}_v\) and \(\eta^{(i)}_k\) so that they converge to \(\sigma_v\) and \(\eta_k\), respectively.}

\begin{addedblock}
	In practice, MCMC samplers are often equipped with an adaptive warm-up
	(or burn-in) stage to improve mixing and move the Markov chain closer to its
	stationary regime \cite{geman_stochastic_1984}. Following this common
	practice, we suggest executing the proposed DiG sampler in two stages:
	\begin{enumerate}[label=\arabic*)]
		\item \emph{Adaptive warm-up stage}: the sampler is first run for a finite
		      number $N_{\mathrm{w}}$ of iterations using annealed parameters. In the \(i\)-th warm-up
		      iteration, we replace the observation-noise level \(\sigma_v\) with
		      \(\sigma_v^{(i)} \in [\sigma_v,\sigma(T)]\), and for each \(k\in\Hcal\), we replace
		      the modeling mismatch level \(\eta_k\) with
		      \(\eta_k^{(i)} \in [\eta_k,\sigma(T)]\). As the warm-up iteration index \(i\)
		      increases, \(\sigma_v^{(i)}\) and \(\eta_k^{(i)}\) are gradually reduced
		      to the target values \(\sigma_v\) and \(\eta_k\),
		      respectively.

		\item \emph{Stationary sampling stage}: after the warm-up stage terminates, the DiG iteration is conducted with fixed target parameters, i.e., \(\sigma_v^{(i)}\equiv\sigma_v\) and \(\eta_k^{(i)}\equiv\eta_k\) for all subsequent iterations \(i>N_{\mathrm{w}}\).
	\end{enumerate}
\end{addedblock}

Intuitively, inflating the noise and mismatch parameters above their nominal values amounts to replacing each conditional distribution in the Gibbs updates (Algorithm~\ref{alg:Gibbs-relax}) with a ``flatter'' one. This typically allows the Markov chain to explore the state space more freely during \deleted{the early iterations}\added{the warm-up stage}, which empirically accelerates convergence. Although annealing makes each update no longer an exact Gibbs step for the target posterior at \deleted{intermediate}\added{warm-up} iterations, as long as we stop annealing within a finite number of iterations, it does not affect the \deleted{consistency}\added{convergence} of the sampler \added{to the relaxed posterior distribution}, as established by Theorem~\ref{thm:consistency} in the next section.

For the DiG sampler, parameter annealing \added{in the warm-up stage} serves an additional practical purpose. Recall that the conditional draw of \(\breve{\sv}_k\) in Algorithm~\ref{alg:DiG} is implemented by simulating a diffusion process \(\bar{\xv}_{k,t}\), whose numerical realization relies on a denoiser
\begin{equation*}
	D^{(k)}_{\theta}\paren{\zv;\eta} \approx \CondExpt{\sv_k}{ \sv_k + \eta \nv =\zv }.
\end{equation*}
When training \(D^{(k)}_{\theta}(\cdot;\eta)\) from data, the input samples are realizations of \({\sv_k+\eta\nv}\), and the target outputs are the corresponding clean samples \(\sv_k\). As a result, the denoiser is typically well-trained on the high-density region of \(p_{\sv_k+\eta\nv}\), while its approximation error may be much larger in low-density regions. However, during the early stages of DiG, the iterates \(\breve{\sv}_k\) are often strongly influenced by initialization, thus the inputs to the diffusion models (i.e., \(\breve{\uv}_k\) and \(\breve{\rv}'_{\neg k}\)) may lie outside the regions frequently seen during training. Feeding such out-of-distribution inputs into the pre-trained denoisers may therefore incur large errors. Parameter annealing mitigates this issue by calling \(D^{(k)}_{\theta}(\cdot;\eta)\) with a larger noise level \(\eta\), which effectively enlarges the high-probability region covered by the training distribution and thus reduces the risk of severe extrapolation errors in the early Gibbs iterations.

For the above two reasons, we recommend incorporating \deleted{parameter annealing}\added{the adaptive warm-up strategy} as a standard ingredient in practical implementations of DiG. \deleted{Guidelines for choosing the annealing schedules will be discussed in the next section.}

\begin{addedblock}
	\paragraph*{Computational and storage cost}
	With the adaptive warm-up strategy, one DiG sweep consists mainly of Gaussian updates of \(\uv_k\) for \(k\in\Hcal\) and reverse-diffusion updates of \(\sv_k\) for all components. Since \(\Hm_k^\top\Hm_k\) is shared across iterations, the Gaussian update can be accelerated by precomputing and caching a dense matrix factorization, such as an SVD or a Cholesky factorization of the Gaussian precision matrices. This gives a preprocessing cost \(O(md_k^2+d_k^3)\) and enables each subsequent Gaussian update to be performed in \(O(md_k+d_k^2)\). Let \(M_k^{(i)}\) be the number of discretization steps in the reverse-diffusion sampler for component \(k\) in sweep \(i\), and let \(F_k\) be the cost of one denoiser evaluation. Then DiG has a one-time preprocessing cost \(O\!\left(\sum_{k\in\Hcal}(md_k^2+d_k^3)\right)\), and one sweep costs
	\begin{equation*}
		O\!\left(
		\sum_{k\in\Hcal}(md_k+d_k^2) + \sum_{k=1}^K M_k^{(i)}F_k
		\right).
	\end{equation*}
	The storage is dominated by the component states, matrix factors, and diffusion-model parameters:
	\begin{equation*}
		O\!\left(
		\sum_{k=1}^K d_k
		+
		\sum_{k\in\Hcal} d_k^2
		+
		\sum_{k=1}^K P_k
		\right),
	\end{equation*}
	where \(P_k\) is the number of parameters of the \(k\)-th diffusion model. When \(K\), \(d_k\), \(F_k\), or \(P_k\) become very large, these computational and storage requirements can become practical limitations of the modular DiG implementation.
\end{addedblock}

By combining the developments in Sections~\ref{sec:prior-modeling} and~\ref{sec:DiG-sampling}, we obtain the proposed
signal component decomposition framework. Specifically, the framework
(i) provides a unified way to incorporate \deleted{both}\added{component-wise} model-driven and data-driven prior knowledge\deleted{ of each component} into diffusion priors, and
(ii) enables plug-and-play use of these diffusion priors within a Gibbs-type sampler.
The resulting diffusion-within-Gibbs (DiG) sampler draws exact samples from the posterior of a perturbed observation model, with a perturbation error governed by user-specified parameters.

\section{Theoretical Analysis and Discussion}
\label{sec:theory}

In this section, we establish the asymptotic consistency of the proposed DiG sampler\added{. We also discuss the practical implications of the perfect diffusion model assumption (Assumption~\ref{assump:diff}) required by the consistency theorem,} and compare DiG with a recently proposed class of diffusion-based posterior samplers built upon certain variable-splitting techniques.\deleted{ We further present some useful tips for the implementation of the DiG algorithm.}

\subsection{Consistency of the DiG Sampler}\label{subsec:consistency}

We first show that, when the diffusion models are trained perfectly, the DiG algorithm asymptotically produces samples from the exact posterior of the perturbed model \eqref{eq:observation-relax}. Specifically, \deleted{under}\added{we consider the DiG sampler equipped with the adaptive warm-up strategy described in Section~\ref{subsec:warm-up}. Under} Assumption~\ref{assump:diff} and \deleted{two additional conditions in Theorem~\ref{thm:consistency}---namely,
	Assumption~(a), which ensures irreducibility and aperiodicity of the induced Markov chain, and Assumption~(b), which ensures
	regularity of the annealing schedule---}\added{a positive prior-density condition, which ensures irreducibility and aperiodicity of the induced Markov chain,} we prove that the distribution of the DiG output converges to the posterior distribution
induced by the relaxed observation model~\eqref{eq:observation-relax}.

\begin{theorem}\label{thm:consistency}
	Under Assumptions~\ref{assump:noise}, \ref{assump:independence}, and \ref{assump:diff}, consider the DiG sampler (Algorithm~\ref{alg:DiG}) implemented with \deleted{parameter annealing}\added{the finite adaptive warm-up strategy (Section~\ref{subsec:warm-up})}.
	\begin{deletedblock}
		Assume further that:
		\begin{enumerate}[label=(\alph*)]
			\item For all \(1\leq k\leq K\), the distribution of \(\sv_k\) admits a density with respect to the Lebesgue measure on \(\R^{d_k}\) that is strictly positive everywhere.

			\item There exists an index \(N_0\) such that \(N_0 < N\) and, for all \(i\geq N_0\), the annealed parameters satisfy \(\sigma_v^{(i)}\equiv \sigma_v\) and, for all \(k\in\Hcal\), \(\eta_k^{(i)}\equiv \eta_k\).
		\end{enumerate}
	\end{deletedblock}
	\added{Suppose that each prior \(p_{\sv_k}\) has a strictly positive Lebesgue density on \(\R^{d_k}\).}
	\deleted{Let the initial samples $\sv^{(0)}_{1:K}$, $\uv^{(0)}_{\Hcal}$ of DiG be arbitrary random variables, and denote the resultant output random variables of DiG after $N$ iterations by \(\sv^{(N)}_{1:K}\), then we have}\added{Let \(\sv^{(N)}_{1:K}\) denote the DiG output after \(N\) iterations. Then}
	\begin{equation*}
		\lim_{N\to \infty} \TV\!\dparen{ p_{{\sv}^{(N)}_{1:K}} }{ p_{\sv_{1:K} \mid \yv'=\breve{\yv}} } = 0,
	\end{equation*}
	i.e., as \(N\) increases, the joint distribution of \(\sv^{(N)}_{1:K}\) converges in total variation distance to the posterior of \eqref{eq:observation-relax}.
\end{theorem}

\begin{IEEEproof}
	See Appendix~\ref{proof:consistency}.
\end{IEEEproof}

Theorem~\ref{thm:consistency} establishes the convergence of DiG to the relaxed posterior distribution. Since the relaxed observation model \eqref{eq:observation-relax} differs from the original model \eqref{eq:observation} only through the relaxation levels \(\eta_k\), the distribution \(p_{\sv^{(N)}_{1:K}}\) provides an accurate approximation to the true posterior \(p_{\sv_{1:K}\mid \yv=\breve{\yv}}\) when \(\eta_k\) is sufficiently small.

\added{One may notice that some structural priors, such as exact sparse or exact low-rank priors, may be singular with respect to the Lebesgue measure and thus do not satisfy the strictly positive density condition in Theorem~\ref{thm:consistency}. This condition is mainly used as a sufficient technical condition that allows standard convergence results to be applied to the DiG sampler. Although it does not apply to certain structural priors exactly, it is mild from an approximation viewpoint: in the quadratic Wasserstein space~\cite{villani_optimal_2009}, any finite-second-moment prior can be approximated arbitrarily well by distributions admitting strictly positive densities, e.g., via Gaussian convolution with vanishing variance. Theorem~\ref{thm:consistency} applies exactly to these positive-density approximating priors.}

In particular, when all sensing operators are identity matrices (one may directly reduce to this case via an appropriate change of variables; see Section~\ref{subsec:formulation}), no modeling mismatch is introduced in \eqref{eq:observation-relax}, and one obtains convergence to the true posterior directly.

\begin{corollary}\label{cor:consistency}
	Under the setting of Theorem~\ref{thm:consistency}, if \(\Hm_k=\Imat\) for all \(1\leq k\leq K\), then the DiG output \({\sv}^{(N)}_{1:K}\) satisfies
	\begin{equation*}
		\lim_{N\to \infty} \TV\!\dparen{ p_{{\sv}^{(N)}_{1:K}} }{ p_{\sv_{1:K} \mid \yv=\breve{\yv}} } = 0,
	\end{equation*}
	i.e., \({\sv}^{(N)}_{1:K}\) converges in distribution to the true posterior.
\end{corollary}

\begin{addedblock}
	\subsection{Discussion on Imperfect Diffusion Models} \label{subsec:discuss-imperfect-DM}
\end{addedblock}

Note that Theorem~\ref{thm:consistency} assumes perfectly trained diffusion models (Assumption~\ref{assump:diff}). Concretely, this idealized assumption entails:
\begin{enumerate}[label=\arabic*)]
	\item For all \(\eta>0\), the denoiser \(D^{(k)}_\theta(\cdot;\eta)\) in the diffusion model exactly matches the conditional expectation \(\CondExpt{\sv_k}{\sv_k+\eta\nv=\cdot}\). This effectively requires either
	      \begin{itemize}
		      \item infinitely many training samples as specified by Assumption~\ref{assump:prior}, a neural network structure expressive enough to represent the conditional expectation function, and training that attains a global optimum of the loss function; or
		      \item an exact analytic prior in Assumption~\ref{assump:prior}, so that \(D^{(k)}_\theta\) is a fully model-driven denoiser computed directly as the MMSE estimator.
	      \end{itemize}
	\item The integration step size \(h\) in the numerical solver \eqref{eq:diff-integral-solver} tends to zero.
\end{enumerate}
This idealization is unattainable in practice and must be approximated: the first requirement is better met by using more training data, more expressive architectures, improved optimization methods, or more accurate analytic priors and denoisers; the second is better met by increasing the number of numerical integration steps. Therefore, Theorem~\ref{thm:consistency} is not intended to directly predict the practical performance of DiG. Nevertheless, it conveys two key messages:
\begin{enumerate}
	\item The error introduced by the algorithmic design of DiG is controlled solely by the relaxation levels \(\eta_k\); when \(\eta_k\) is sufficiently small, this design-induced error becomes negligible.
	\item If the empirical performance of DiG falls short of expectations, then (under the modeling assumptions) the only remaining source of error is the diffusion model itself, which helps pinpoint where debugging and performance improvements should focus.
\end{enumerate}

\begin{addedblock}
	One may also be interested in the robustness of the DiG sampler to
	imperfect diffusion models. More precisely, approximation errors in the
	diffusion models perturb the transition kernel of the Gibbs sampler, and one may ask whether this
	perturbation can lead to uncontrolled degradation of the decomposition
	performance of DiG. Existing perturbation analyses for approximate
	Markov chains provide possible theoretical routes
	\cite{mitrophanov_stability_2003,rudolf_perturbation_2018}, but applying
	them to concrete diffusion-based samplers typically requires ergodicity
	assumptions or Lyapunov-type conditions that are difficult to verify.
	Moreover, the resulting error bounds often involve constants that are hard
	to estimate or interpret quantitatively. We therefore examine this
	robustness issue empirically rather than claiming a general theoretical
	guarantee; see Section~\ref{subsec:heartbeat-extraction}.
\end{addedblock}

\subsection{Comparison with Other Diffusion-Based Samplers}\label{subsec:compare-samplers}

Recently, a class of diffusion-based posterior samplers \cite{coeurdoux_pnp_2024,xu2024,wu_principled_2024} built upon certain variable-splitting techniques has been proposed for solving general Bayesian inverse problems. These methods are also asymptotically consistent. However, when applied to component decomposition, they typically \deleted{treat all signal components jointly as a single unknown}\added{perform posterior sampling over all signal components jointly
	as a single aggregate variable, rather than explicitly exploiting the
	conditional decomposition structure in \eqref{eq:observation}}, and thus cannot fully exploit the structural property of the component decomposition problem\deleted{ in \eqref{eq:observation}}.

We briefly outline how variable splitting can be applied to signal decomposition. Define the stacked variable
\begin{align*}
	\sv & \coloneq \begin{bmatrix}
		               \sv_1^\top & \sv_2^\top & \cdots & \sv_K^\top
	               \end{bmatrix}^\top, \\
	\Hm & \coloneq \begin{bmatrix}
		               \Hm_1 & \Hm_2 & \cdots & \Hm_K
	               \end{bmatrix},
\end{align*}
so that \eqref{eq:observation} is rewritten as a linear inverse problem
\(\yv = \Hm \sv + \vv\).
By Bayes' rule, the posterior density factorizes as
\begin{align*}
	p_{\sv\mid \yv}\cparen{\breve{\sv}}{\breve{\yv}}
	 & \propto  p_{\yv\mid\sv}\cparen{\breve{\yv}}{\breve{\sv}} \times p_{\sv}\paren{\breve{\sv}}                           \\
	 & \propto \exp\!\paren{-\inv{2\sigma_v^2}\big\|\breve{\yv}-\Hm\breve{\sv}\big\|^2_2}\times p_{\sv}\paren{\breve{\sv}}.
\end{align*}
Some plug-and-play (PnP) diffusion samplers \cite{bouman_generative_2023,xu2024} split the sampling of the above product distribution into \emph{proximal} sampling \cite{lee_structured_2021} steps\footnote{
	For a distribution \(p\), the \(\eta\)-proximal sampler at \(\xv\) draws a sample from a density proportional to
	\(p(\cdot)\exp\!\big(-\| \cdot-\xv\|_2^2/(2\eta^2)\big)\).
} of the two factors therein. Concretely, their iterations take the form
\begin{align*}
	\text{Draw}~\sv^{(i+\hf)} & \sim C \exp\!\paren{
		-\frac{\big\|\breve{\yv}-\Hm(\cdot)\big\|^2_2}{2\sigma_v^2}
	-\frac{\big\|(\cdot)-\sv^{(i)}\big\|^2_2}{2\eta^2}},                    \\
	\text{Draw}~\sv^{(i+1)}   & \sim C'\, p_{\sv}(\cdot)\times\exp\!\paren{
		-\inv{2\eta^2}\big\|(\cdot)-\sv^{(i+\hf)}\big\|^2_2 },
\end{align*}
where \(C\) and \(C'\) are normalization constants and \(\eta>0\) is chosen sufficiently small. In the above iteration, \(\sv^{(i+\hf)}\) is drawn by sampling from a Gaussian distribution, whereas \(\sv^{(i+1)}\) is drawn using a diffusion model.

It has been pointed out \cite{wu_principled_2024} that the aforementioned proximal sampling scheme is equivalent to Gibbs sampling on the following joint distribution of \((\sv,\uv)\):
\begin{equation*}
	p_{\sv,\uv}(\breve{\sv},\breve{\uv})
	\propto
	p_{\sv}\paren{\breve{\sv}}
	\times
	\exp\!\paren{
		-\frac{\big\|\breve{\yv}-\Hm\breve{\uv}\big\|^2_2}{2\sigma_v^2}
		-\frac{\big\|\breve{\uv}-\breve{\sv}\big\|^2_2}{2\eta^2} },
\end{equation*}
with $\paren{\sv^{(i)},\sv^{(i+\hf)}}$ being the $i$-th generated sample for $(\sv,\uv)$. One can further verify that this joint distribution is precisely the posterior $p_{\sv,\uv|\yv''}\cparen{\cdot}{\breve{\yv}}$ associated with the following relaxed observation model:
\begin{equation}
	\yv'' \coloneq \sum_{k=1}^K \Hm_k \paren{ \sv_k + \vv_k } + \vv,
\end{equation}
where \(\uv_k\coloneq \sv_k+\vv_k\) and \(\vv_k\sim\Normal{\zerov}{\eta^2\Imat}\). This observation reveals two key differences between DiG and the above proximal-sampling approach, which are confirmed by the numerical results in Section~\ref{sec:experiments}:
\begin{enumerate}
	\item Each iteration of the proximal-sampling method updates the entire stacked variable \(\sv\equiv \sv_{1:K}\) simultaneously, while DiG exploits component independence and updates the components sequentially within one sweep. Empirically, this often yields faster convergence for DiG.

	\item DiG introduces modeling mismatch only for components with non-identity sensing operators, whereas the proximal-sampling method injects the same mismatch level into all components. When some \(\Hm_k\) are identity matrices, this leads to unnecessary mismatch and typically degrades sampling accuracy compared with DiG.
\end{enumerate}

Finally, we note that the relaxed posterior targeted by DiG is also derivable from a more general variable-splitting framework introduced in \cite{vono_split-and-augmented_2019}, provided that one suitably reformulates $p_{\sv_{1:K}|\yv}$ and introduces appropriate relaxation factors. However, the underlying derivations differ: the variable splitting in \cite{vono_split-and-augmented_2019} is motivated by approximating the target posterior density function, whereas our relaxation is obtained by explicitly constructing a new probabilistic model. Moreover, though the variable-splitting scheme in \cite{vono_split-and-augmented_2019} can yield the same relaxed posterior as in this paper, the split Gibbs method in \cite{vono_split-and-augmented_2019} still updates all components synchronously at each iteration, and therefore does not fully leverage the component-wise independence.

\begin{deletedblock}
	\subsection{Implementation Details}
	\label{sec:implementation}

	Although DiG comes with asymptotic consistency guarantees, current theory provides limited insight into its mixing rate. In practice, the choice of several algorithmic parameters often has a pronounced impact on mixing time. While we still lack a complete theoretical understanding of how these parameters affect mixing properties, a few empirical heuristics have shown strong performance in our experiments. We summarize them below.

	\paragraph*{1) Choice of initial samples} DiG typically converges faster when the initial point \(\sv^{(0)}_{1:K}\) lies in a high-density region of the posterior. In practice, if estimates of the prior mean \(\xiv_k\approx \Expt{\sv_k}\) and covariance \(\Gammam_k\approx \Cov\paren{\sv_k}\) of each component are available—either from an analytic prior or from data—then one may approximate \(p_{\sv_k}\) by a Gaussian distribution \(\Normal{\xiv_k}{\Gammam_k}\), and use the resulting linear-Gaussian model to form an approximate posterior mean:
	\begin{equation*}
		\CondExpt{\sv_k}{\yv}\approx \xiv_k + \Gammam_k \Hm_k^\top {\Gammam_{\yv}}^{-1}
		\paren{ \yv - \sum_{k=1}^K \Hm_k \xiv_k },
	\end{equation*}
	where \(\Gammam_{\yv}\coloneq \paren{ \sum_{k=1}^K \Hm_k\Gammam_k\Hm_k^\top + \sigma_v^2\Imat }\).
	We recommend using this approximate posterior mean as the initialization \(\sv^{(0)}_k\).

	\paragraph*{2) Annealing schedule design} In the annealing strategy, the decay profiles of \(\paren{\sigma_v^{(i)}}_{i=1}^{N}\) and \(\paren{\eta_k^{(i)}}_{i=1}^{N}\) toward \(\sigma_v\) and \(\eta_k\), respectively, constitute an important user-controlled design choice. We recommend adopting a cosine-shaped schedule inspired by the noise schedules commonly used in diffusion models \cite{nichol_improved_2021}, i.e.,
	\begin{align*}
		\sigma_v^{(i)} & \coloneq \sigma_{\min}
		+ \paren{\sigma_{\max} - \sigma_{\min}}\times
		{ \cos^2 \!\paren{ \frac{i/N+\epsilon}{1+\epsilon}\cdot \frac{\pi}{2} } }, \\
		\eta_k^{(i)}   & \coloneq \eta_{\min}^{(k)}
		+ \paren{\eta_{\max}^{(k)} - \eta_{\min}^{(k)}}\times
		{ \cos^2 \!\paren{ \frac{i/N+\epsilon}{1+\epsilon}\cdot \frac{\pi}{2} } },
	\end{align*}
	where $\epsilon$ is a small constant. This schedule is relatively flat during the early iterations (where initialization effects are strong) and the late iterations (which largely determine final sample quality), while decaying more rapidly—approximately linearly—during the middle phase.

\end{deletedblock}

\section{Experiments}
\label{sec:experiments}

This section evaluates the proposed component decomposition framework through \deleted{two numerical experiments}\added{two experimental settings designed to examine different aspects of the proposed framework}.

In the first experiment (an illustrative multi-component image decomposition example), we benchmark DiG against the
state-of-the-art variable-splitting-based diffusion sampler DPnP \cite{xu2024}.
\added{We also investigate how practical implementation choices, including the adaptive warm-up strategy and initialization, affect the finite-iteration behavior of DiG.}

In the second experiment (heartbeat extraction task under strong motion interference), we compare DiG with a broader set of baselines\deleted{, including two classical decomposition methods (EMD \cite{huang1998} and VMD \cite{dragomiretskiy2014}) and two diffusion-based samplers	(DPnP \cite{xu2024} and MSDM \cite{mariani2024}).}
\added{. These baselines include three unsupervised methods: EMD \cite{huang1998}, VMD \cite{dragomiretskiy2014}, and a Gaussian-process-based (GP) Bayesian decomposition method \cite{Duvenaud2011}. They also include three state-of-the-art diffusion-based posterior samplers: MSDM \cite{mariani2024}, DPnP \cite{xu2024}, and DAPS \cite{zhang2025cvpr}.}
\deleted{Notably, in the second experiment, when combined with appropriate model-based priors, DiG achieves better decomposition quality than these competing methods, while requiring substantially less training data.}
\added{In this task, we further examine three practical aspects of DiG: its performance degradation when the learned diffusion prior becomes less accurate under limited training data; the effectiveness of hybrid prior modeling (Section~\ref{subsec:train_diff_model}), which combines learned and model-based component priors; and uncertainty quantification based on posterior samples generated by DiG.}

\subsection{Image Decomposition with Corrected Components}\label{subsec:image-decomposition}

We first evaluate the effectiveness of the DiG algorithm in a highly underdetermined image decomposition problem. The observation model is
\begin{equation}\label{eq:image-decomposition}
	\yv = \sv_1 + \Hm_2 \sv_2 + \Hm_3 \sv_3 + \Hm_4 \sv_4 + \vv,
\end{equation}
where $\sv_1\in\R^{28\times 28}$ is an image of clothes drawn from the Fashion-MNIST dataset \cite{xiao2017fashion}, $\sv_2\in\R^{28\times 28}$ and $\sv_3\in\R^{28\times 28}$ are images of single handwritten digits drawn from the MNIST dataset \cite{dengmnist} with labels (values) satisfying
\begin{equation}\label{eq:label-constraint}
	\text{label}(\sv_2) + 1 \equiv \text{label}(\sv_3) \pmod{10},
\end{equation}
and $\sv_4\in\R^{28\times 28}$ is an image of a single English letter drawn from the EMNIST dataset \cite{cohen2017emnist}. All component images are converted to grayscale, and are normalized to the range $[0,1]$. The linear operator $\Hm_2$ applies a $90^\circ$ clockwise rotation and then applies a sign flip, whilst $\Hm_3$ and $\Hm_4$ respectively represent an embossing transform and a Gaussian blur operator. The standard deviation of the observation noise $\vv$ is \deleted{$\sigma_v = 0.05$}\added{$\sigma_v=0.25$}.

Notice that \eqref{eq:label-constraint} enforces certain dependence relation between $\sv_2$ and $\sv_3$. Thus we define a new active component $\sv'\coloneq \begin{bmatrix}
		\sv^\top_2 & \sv_3^\top
	\end{bmatrix}^{\top}$ and sample from $p_{\sv_1,\sv',\sv_4|\yv}$ instead.

\begin{table}[t]
	\centering
	\caption{Relative squared error of all algorithms. \textbf{Bold}: best. }
	\begin{tabular}{lcccc}
		\hline
		\textbf{Method} & ${\text{RSE}(\sv_1)}$ & ${\text{RSE}(\sv_2)}$ & ${\text{RSE}(\sv_3)}$ & ${\text{RSE}(\sv_4)}$ \\
		\hline
		DiG (w/ c)      & \textbf{0.095}        & \textbf{0.192}        & \textbf{0.043}        & \textbf{0.118}        \\
		DPnP (w/ c)     & 0.114                 & 0.215                 & 0.046                 & 0.123                 \\
		DiG (w/o c)     & 0.101                 & 0.279                 & 0.048                 & 0.146                 \\
		DPnP (w/o c)    & 0.115                 & 0.281                 & 0.047                 & 0.133                 \\
		\hline
	\end{tabular}
	\label{tab:relative_error}
\end{table}

\begin{figure}[t]
	\centering
	\includegraphics[width=.95\linewidth]{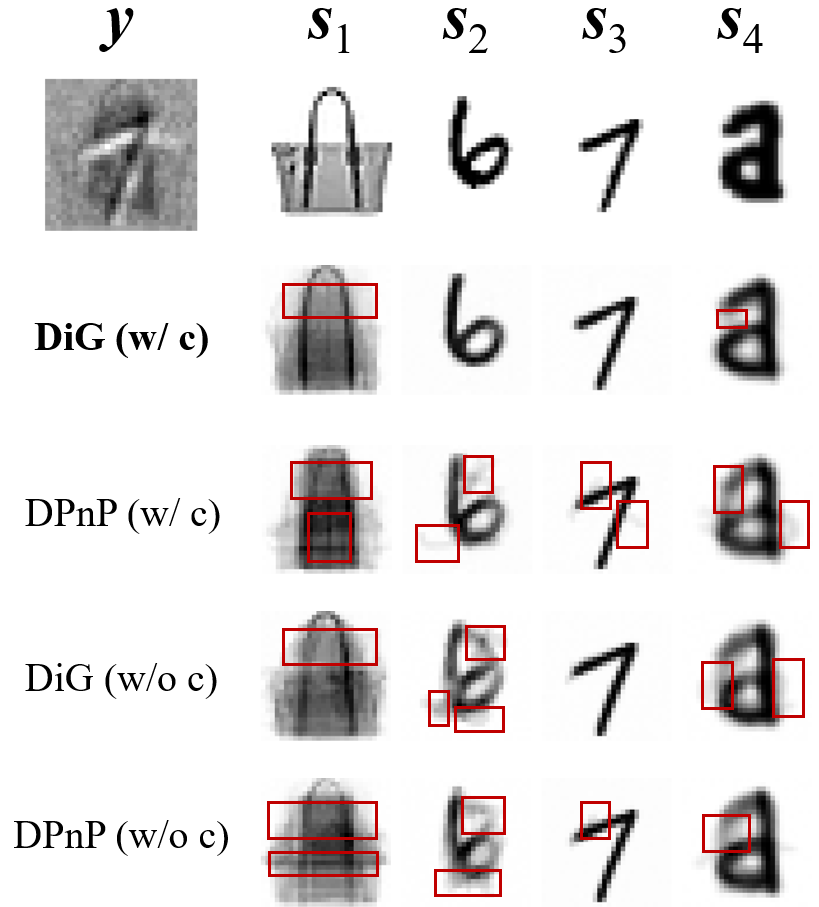}
	\caption{Estimates given by all algorithms in a representative problem instance. First row: the observation and ground-truth of all component signals. Other rows: estimates given by each algorithm. The red boxes highlight representative artifact regions in the estimates produced by each method.}
	\label{fig:mnist_mean}
\end{figure}

\paragraph*{1) Algorithms for comparison}  We compare DiG (Algorithm~\ref{alg:DiG}) against \deleted{a state-of-the-art diffusion-based sampler termed DPnP \cite{xu2024}, where the latter treats all the components as a single aggregate unknown}\added{DPnP \cite{xu2024}, a state-of-the-art variable-splitting-based diffusion posterior sampler} (cf. Section~\ref{subsec:compare-samplers}). \added{Unlike DiG, DPnP performs variable splitting on the stacked component vector, thereby injecting the same relaxation mismatch into all components, including those with identity sensing operators, and updates the components jointly rather than component-wise.} To assess the role of modeling the correlation between $\sv_2$ and $\sv_3$, we include \deleted{ablation studies where}\added{comparison variants in which} DiG and DPnP treat $\sv_2$ and $\sv_3$ as if they are independent.\footnote{Specifically, when the correlation between $\sv_2$ and $\sv_3$ is considered, DiG and DPnP employ a single denoiser $D'_\theta(\breve{\sv}_2,\breve{\sv}_3;\eta)$ to jointly denoise $\sv_2$ and $\sv_3$ in the diffusion model for sampling $\sv'$. When $\sv_2$ and $\sv_3$ are treated as independent, two separate denoisers $D_\theta^{(2)}(\breve{\sv}_2;\eta)$ and $D_\theta^{(3)}(\breve{\sv}_3;\eta)$ are used respectively for $\sv_2$ and $\sv_3$ in the diffusion model. \label{fm:correlation-treatment}} We denote an algorithm (e.g., DiG) that takes correlation into account as, e.g., DiG (w/ c), and denote the one that neglects correlation as, e.g., DiG (w/o c).

\paragraph*{\added{2) Implementation settings}}

\begin{deletedblock}
	For both DiG and DPnP, we set the initial sample as suggested in Section~\ref{sec:implementation} with $\xiv_k\coloneq \zerov$ and $\Gammam_k\coloneq 0.04 \Imat$ for all $k$. We set the parameter annealing schedule as suggested in Section~\ref{sec:implementation} with $\sigma_{\max}\coloneq 20\sigma_v \eqcolon \eta^{(k)}_{\max}$ ($k=2,3,4$) and $\sigma_{\min}\coloneq 5\sigma_v \eqcolon \eta^{(k)}_{\min}$ ($k=2,3,4$). The number of iterations is fixed to $N=100$. In each iteration, to sample from $p_{\bar{\xv}_{k,0}|\bar{\xv}_{k,\sigma^{-1}(\eta)}}$, we run the corresponding reverse diffusion from
	\(t=\sigma^{-1}(\eta)\) to \(t=0\) using 100 discretization (numerical integration) steps. For each component signal, we draw $25$ posterior samples and compute an estimate of the component by taking the average of posterior samples whose data misfit is lower than the average misfit value.
\end{deletedblock}
\begin{addedblock}
	For both DiG and DPnP, we initialize the component estimates using a simple linear-Gaussian surrogate model. Specifically, we approximate each component prior by a Gaussian distribution with mean $\xiv_k\coloneq \zerov$ and covariance $\Gammam_k\coloneq 0.04\Imat$, compute the corresponding closed-form MMSE estimate under this surrogate model, and use it as the initial sample. We run an adaptive warm-up stage for the first $N_{\mathrm{w}}=150$ iterations. During warm-up, the inflated observation-noise level and mismatch levels are decreased according to
	\begin{align*}
		\sigma_v^{(i)}
		 & \coloneq \sigma_{v,\min}
		+ \paren{\sigma_{v,\max}-\sigma_{v,\min}}
		\cos^2\!\paren{\frac{i/N_{\mathrm{w}}+\epsilon}{1+\epsilon}\cdot\frac{\pi}{2}}, \\
		\eta_k^{(i)}
		 & \coloneq \eta_{\min}^{(k)}
		+ \paren{\eta_{\max}^{(k)}-\eta_{\min}^{(k)}}
		\cos^2\!\paren{\frac{i/N_{\mathrm{w}}+\epsilon}{1+\epsilon}\cdot\frac{\pi}{2}},
	\end{align*}
	where $\sigma_{v,\max}=20\sigma_v$ and $\sigma_{v,\min}=\sigma_v$. For the mismatch levels, we use $\eta_{\max}^{(k)}=20\sigma_v$ and $\eta_{\min}^{(k)}=\sigma_v$ for $k=2,3,4$ in DiG, and for $k=1,\dots,4$ in DPnP. After warm-up, the remaining iterations use the fixed target parameters $\sigma_v^{(i)}=\sigma_v$ and $\eta_k^{(i)}=\sigma_v$, with the same component-index convention. The number of iterations is fixed to $N=500$.

	In each iteration, to sample from $p_{\bar{\xv}_{k,0}|\bar{\xv}_{k,\sigma^{-1}(\eta)}}$, we run the corresponding reverse diffusion from
	\(t=\sigma^{-1}(\eta)\) to \(t=0\) using 100 discretization (numerical integration) steps. For each component signal, we draw $25$ posterior samples and compute an estimate of the component by taking the average of posterior samples whose data misfit is lower than the average misfit value.
\end{addedblock}

\deletedparagraph{2) Diffusion models} \added{The diffusion-model training settings used in this experiment are as follows.} We adopt standard U-net-style denoisers in the diffusion models for sampling $\sv_1$, $\sv'\coloneq \begin{bmatrix}
		\sv^\top_2 & \sv_3^\top
	\end{bmatrix}^{\top}$ and $\sv_4$, as suggested in \cite{song2021}. We set the noise schedule $g(t)$ in the diffusion model (cf. Assumption~\ref{assump:diff}) as $g(t)\coloneq
	\alpha^t$ ($t\in[0,T]$) with $\alpha=15$ and $T=1$. We employ the same denoising score-matching loss as \cite{song2021}, and adopt an exponentially decaying learning rate initialized at $8\times 10^{-5}$. The number of training samples is 100 K for training the denoiser of $\sv'$, and 60 K for other denoisers. Each denoiser is trained for $500$ epochs. The reverse diffusion process \eqref{eq:reverse-SDE-sk} for generating samples is computed numerically using a first-order SDE integrator \cite{song2021}.

\paragraph*{3) \deleted{Results}\added{Main results}} For evaluation, we generate 200 synthetic observations $\yv$ in \eqref{eq:image-decomposition} using data from the test dataset of Fashion-MNIST, MNIST and EMNIST. For each algorithm, we quantify its reconstruction accuracy by the relative squared error (RSE)
\begin{equation}\label{eq:error_signal}
	\text{RSE}(\sv_k)
	= \frac{\Expt{\|\hat{\sv}_k - \sv_k\|_2^2}}{\Expt{\|\sv_k\|_2^2}},
	\qquad k = 1,\dots,4,
\end{equation}
where $\hat{\sv}_k$ is the estimate of $\sv_k$ given by the algorithm, and expectation is taken over all 200 problem instances. The RSE of all methods are reported in Table~\ref{tab:relative_error}. Fig.~\ref{fig:mnist_mean} shows the estimates given by all algorithms in a representative problem instance. Table~\ref{tab:relative_error} together with Fig.~\ref{fig:mnist_mean} indicate that DiG (w/ c) achieves the most accurate estimation for all component signals, and that exploiting the correlation between $\sv_2$ and $\sv_3$ with a joint diffusion model (or say, joint denoiser) leads to a noticeable performance improvement.

\begin{addedblock}
	Table~\ref{tab:image-runtime} reports the wall-clock inference time of all methods in this experiment, averaged over 50 test instances. The four columns report the total time per test instance (including the matrix-factorization precomputation required by the Gaussian updates), the time spent in Gaussian updates, the time spent in reverse-diffusion sampling, and the average time per iteration, respectively. From Table~\ref{tab:image-runtime}, DiG and DPnP have comparable total inference time under the same posterior-sample budget. For both methods, reverse-diffusion sampling dominates the total cost, whereas Gaussian updates account for only a small fraction. Modeling the correlation between $\sv_2$ and $\sv_3$ further reduces the runtime in both DiG and DPnP, because the correlated implementation uses one joint denoiser for $(\sv_2,\sv_3)$ instead of two separate denoisers.
\end{addedblock}

\begin{table}[t]
	\centering
	\caption{\added{Wall-clock inference time in the image decomposition experiment. All entries are measured in seconds.}}
	\begin{tabular}{lcccc}
		\toprule
		\textbf{Method} & \textbf{Total/instance} & \textbf{Gaussian} & \textbf{Diffusion} & \textbf{Avg./iter.} \\
		\midrule
		DiG (w/ c)      & 190.450                 & 1.201             & 185.660            & 0.952               \\
		DPnP (w/ c)     & 188.709                 & 1.640             & 184.737            & 0.944               \\
		DiG (w/o c)     & 253.417                 & 1.265             & 239.621            & 1.267               \\
		DPnP (w/o c)    & 253.816                 & 1.619             & 248.881            & 1.269               \\
		\bottomrule
	\end{tabular}
	\label{tab:image-runtime}
\end{table}

\begin{addedblock}

	\paragraph*{4) Ablation studies}

	We further investigate how two practical implementation choices, namely the adaptive warm-up strategy and the initialization strategy, affect the finite-iteration behavior of DiG. We also compare DiG with DPnP under the same implementation choices. All ablation studies use the correlation-aware variants of DiG and DPnP. In each ablation study, we use the same experimental settings as before, and vary only the implementation choice being ablated. We report the mean component RSE, i.e., the average of \(\text{RSE}(\sv_k)\) over the four components, as a function of the iteration number.

	Fig.~\ref{fig:image-warmup-ablation} studies how the presence of the adaptive warm-up strategy and the specific warm-up schedule affect the finite-iteration behavior of DiG and DPnP. We compare three choices: no warm-up, the cosine warm-up schedule used above, and a linear warm-up schedule. The results show that, under each warm-up implementation, DiG consistently achieves a faster decrease and a lower final value in mean component RSE than DPnP. For both DiG and DPnP, using warm-up leads to substantially faster error reduction than using no warm-up, demonstrating the practical effectiveness of the warm-up strategy. Among the two warm-up schedules, the cosine and linear designs lead to similar finite-iteration behavior, indicating that the presence of warm-up is more important than the specific schedule shape in this experiment.
\end{addedblock}

\begin{figure}[t]
	\centering
	\includegraphics[width=\linewidth]{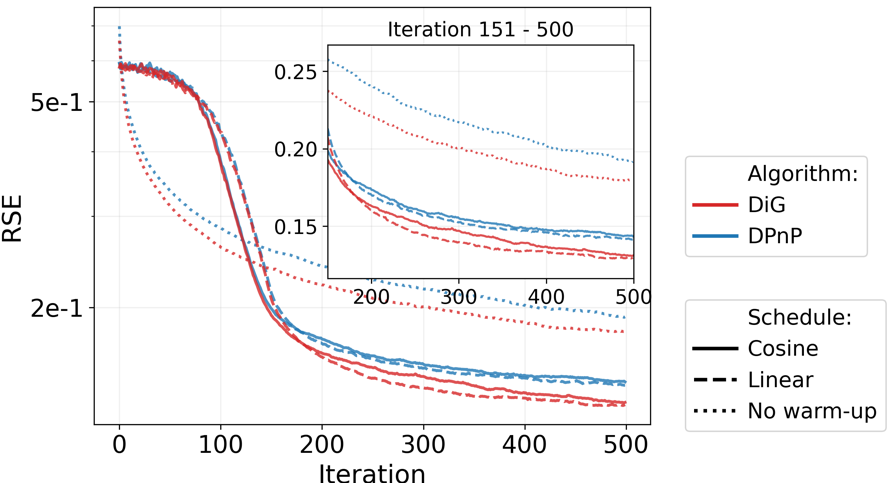}
	\caption{\added{Warm-up schedule ablation in the image decomposition experiment. Colors distinguish DiG and DPnP, while line styles distinguish different warm-up implementations. For the cosine and linear warm-up variants, the warm-up stage lasts for $N_{\mathrm{w}}=150$ iterations.}}
	\label{fig:image-warmup-ablation}
\end{figure}

\begin{addedblock}
	Fig.~\ref{fig:image-initialization-ablation} studies how initialization affects the finite-iteration behavior of DiG and DPnP, while fixing the warm-up schedule to the cosine schedule. We compare the Gaussian-surrogate MMSE initialization described above with two random alternatives, denoted by Random 1x and Random 5x, which initialize each component by a zero-mean Gaussian vector with one and five times the component's true variance, respectively. Across all initialization strategies, DiG consistently decreases faster in the stationary sampling stage and reaches a lower final mean component RSE than DPnP. Although different initialization strategies produce different initial RSE values, these differences are quickly reduced under the adaptive warm-up strategy, and the curves behave similarly in the stationary sampling stage. Overall, the initialization strategy has only a limited effect on the final performance in this experiment.
\end{addedblock}

\begin{figure}[t]
	\centering
	\includegraphics[width=\linewidth]{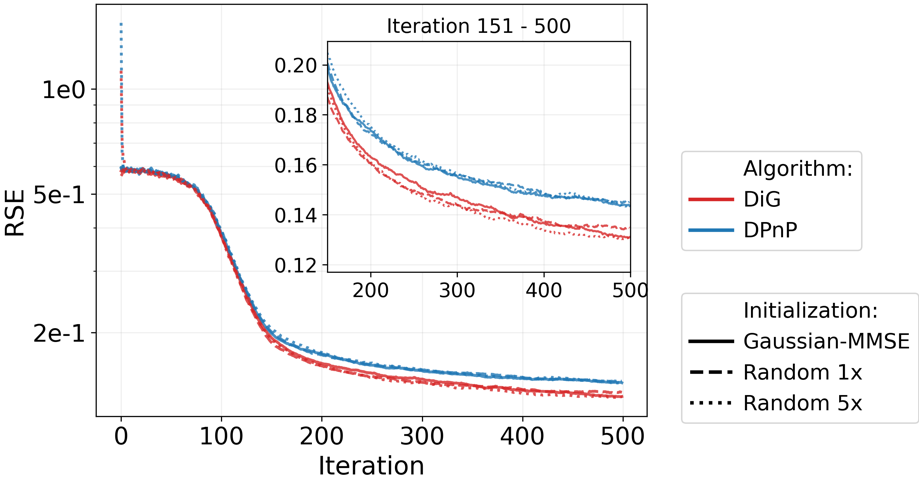}
	\caption{\added{Initialization ablation in the image decomposition experiment under cosine warm-up. Colors distinguish DiG and DPnP, while line styles distinguish different initialization strategies. The inset zooms in on the stationary sampling stage after the first $N_{\mathrm{w}}=150$ warm-up iterations.}}
	\label{fig:image-initialization-ablation}
\end{figure}

\subsection{Heartbeat Extraction from Strong Motion Interference}\label{subsec:heartbeat-extraction}

Next, we consider a more realistic task of extracting heartbeat signals from strong interference caused by body motion, which is inspired by challenges encountered in radar-based contactless heartbeat monitoring \cite{zhang2023overview}. We evaluate the effectiveness of both the DiG sampler (Algorithm~\ref{alg:DiG}) and the hybrid prior modeling techniques (Section~\ref{subsec:train_diff_model}) on this task.

The observation model is
\begin{equation}\label{eq:heartbeat-extraction}
	\yv = \sv_1 + \sv_2 + \vv,
\end{equation}
where $\sv_1\in\R^{1000}$ is the waveform of a heartbeat signal in a fixed time period, $\sv_2\in\R^{1000}$ is an interference signal caused by body motion which is independent from $\sv_1$. The resultant component decomposition problem seems simple at first sight because only two components are involved, and all sensing matrices are identity matrices. However, in radar-based heartbeat monitoring, the amplitude of the heartbeat signal $\sv_1$ is typically much weaker than the interference $\sv_2$, and is comparable to the noise power, which poses a crucial challenge in the extraction of $\sv_1$.

\deletedparagraph{1) Datasets} We use the impedance dataset \cite{schellenberger2020dataset} to generate heartbeat signals for training diffusion models and evaluating all algorithms.
The impedance dataset contains \added{measured} recordings of heartbeat signals from 30 subjects, where impedance (heartbeat) signals from 25 subjects are used for training, and those from the remaining 5 subjects are held out for testing.
The heartbeat signals are bandpass-filtered and segmented into 10-second clips, with the length of each clip being 1000 (sampling interval 0.01 second).
The motion interference components are synthetically generated by integrating 10-second velocity profiles with randomized piecewise-constant amplitudes and smooth sigmoidal transitions\added{, which provides ground-truth component signals for controlled quantitative evaluation}.

\begin{table}[!t]
	\centering
	\setlength{\tabcolsep}{2pt}
	\renewcommand{\arraystretch}{1}
	\caption{\deleted{Relative squared error of recovered heartbeat under different combinations of SIR and SNR in dB. \textbf{Bold}: best.}\added{Relative squared error of recovered heartbeat under different SIR--SNR settings in dB. \textbf{Bold}: best.}}
	\vspace{2mm}
	\resizebox{\columnwidth}{!}{%
		\begin{tabular}{c|ccccccc}
			\hline
			(SIR, SNR)    & EMD       & VMD     & GP     & MSDM           & DPnP  & DAPS  & DiG            \\
			\hline
			(-20.1, 13.2) & 102.671   & 0.116   & 3.595  & \textbf{0.039} & 0.076 & 0.053 & 0.040          \\
			(-20.1, -0.8) & 103.113   & 0.637   & 0.860  & 0.711          & 0.142 & 0.103 & \textbf{0.088} \\
			(-20.1, -6.8) & 104.410   & 2.024   & 0.403  & 3.500          & 0.256 & 0.233 & \textbf{0.147} \\
			(-26.1, 13.2) & 410.347   & 1.231   & 1.443  & 0.171          & 0.141 & 0.107 & \textbf{0.075} \\
			(-26.1, -0.8) & 410.739   & 0.906   & 0.663  & 0.875          & 0.183 & 0.124 & \textbf{0.110} \\
			(-26.1, -6.8) & 412.155   & 1.974   & 0.757  & 3.750          & 0.307 & 0.244 & \textbf{0.172} \\
			(-40.1, 13.2) & 9872.134  & 815.401 & 14.062 & 39.995         & 0.861 & 2.168 & \textbf{0.451} \\
			(-40.1, -0.8) & 10073.719 & 601.728 & 2.509  & 40.961         & 0.705 & 1.109 & \textbf{0.393} \\
			(-40.1, -6.8) & 10076.855 & 122.723 & 4.509  & 44.822         & 0.746 & 0.741 & \textbf{0.481} \\
			\hline
		\end{tabular}}
	\label{tab2}
\end{table}

\paragraph*{\deleted{2) }\added{1) }Algorithms for comparison}

We compare the proposed DiG algorithm with \deleted{four methods:}\added{three unsupervised baselines and three learned diffusion posterior-sampling baselines.}

\added{The unsupervised baselines are} empirical mode decomposition (EMD \cite{huang1998}), variational mode decomposition (VMD \cite{dragomiretskiy2014}), \added{and a Gaussian-process-based (GP) Bayesian decomposition method \cite{Duvenaud2011}.} EMD and VMD are respectively transform-based and optimization-based model-driven methods that decompose an observation as sum of multiple \deleted{periodic components}\added{oscillatory modes}. Both methods are \added{computationally efficient, require no training data, and are} widely used in radar-based vital sign monitoring \cite{zhang2023overview}. To reconstruct the heartbeat signal $\sv_1$ from the \deleted{periodic components given by}\added{modes extracted by} EMD/VMD, we select a subset of components whose sum best approximates the true value of $\sv_1$ in $\ell_2$-distance.\footnote{This oracle selection assumes access to the true heartbeat signal; in practical scenarios without ground truth, the performance of EMD/VMD is expected to degrade.} \added{The GP baseline uses more problem-specific Bayesian modeling assumptions: it models the observation as the sum of a periodic Gaussian-process component and a smooth Gaussian-process component, and performs decomposition by solving a nonconvex optimization problem over the kernel hyperparameters.}

\added{The learned diffusion posterior-sampling baselines are} MSDM \cite{mariani2024}, \deleted{and} DPnP \cite{xu2024}\added{, and DAPS \cite{zhang2025cvpr}}. \deleted{MSDM and DPnP are state-of-the-art posterior sampling methods using plug-and-play diffusion priors.}\added{MSDM is a decomposition method specialized to identity component sensing operators, which generates approximate posterior samples by perturbing the reverse diffusion process of a learned source prior with a likelihood-gradient term. DPnP is a variable-splitting-based sampler that introduces relaxation error even for components with identity sensing operators; see Section~\ref{subsec:compare-samplers}. DAPS alternates reverse-diffusion prior denoising with likelihood-informed Langevin correction along a decoupled noise-annealing schedule; unlike DiG, the reverse-diffusion step is not derived from a component-wise Gibbs conditional distribution. From a theoretical perspective, MSDM and DAPS do not provide any consistency guarantee even under the perfect diffusion model assumption, whereas DPnP is consistent only for a relaxed posterior and therefore lacks DiG's direct consistency guarantee to the true posterior in this identity-operator setting (Corollary~\ref{cor:consistency}).}

\begin{figure}[!t]
	\centering

	\includegraphics[width=.95\linewidth]{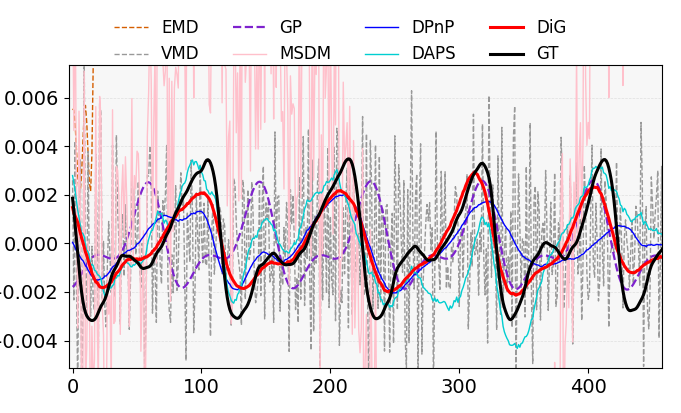}

	\caption{\deleted{Estimates of heartbeat signals given by all algorithms in representative problem instances with different SIR--SNR settings. GT: ground-truth. }\added{Representative heartbeat recovery results under the most challenging setting, $(\mathrm{SIR},\mathrm{SNR})=(-40.1,-6.8)$ dB. GT: ground truth.}}
	\label{fig1}
\end{figure}

\paragraph*{\added{2) Implementation settings}} For MSDM, we set all hyperparameters as suggested in \cite{mariani2024}. \deleted{For DPnP and DiG, we set the initial samples as $\sv^{(0)}_1\coloneq \zerov$ and $\sv^{(0)}_2\coloneq \yv$,}\added{For DPnP and DiG, we initialize the heartbeat component as zero and initialize the motion-interference component by smoothing the observation. For DAPS, we use the same deterministic base initialization and then add the initial annealing noise prescribed by the DAPS algorithm.} \deleted{set the parameter annealing schedule as suggested in Section~\ref{sec:implementation}}\added{Regarding the warm-up strategy, for DiG and DPnP, we use cosine-shaped annealing schedules} with $\sigma_{\max}\coloneq 3\sigma_v$ and $\sigma_{\min}\coloneq \sigma_v$, \added{and run five adaptive warm-up iterations followed by five stationary sampling iterations. For DAPS, we use the same cosine-shaped schedule form with the terminal noise level selected empirically following the DAPS implementation and run the same number of sampling iterations.} \deleted{and fix the number of iterations to $N=5$.} In each iteration, to sample from $p_{\bar{\xv}_{k,0}|\bar{\xv}_{k,\sigma^{-1}(\eta)}}$, we run the corresponding reverse diffusion from
\(t=\sigma^{-1}(\eta)\) to \(t=0\) using 100 discretization (numerical integration) steps. For all sampling-based methods, we draw 25 posterior samples for each problem instance and use their average as the final estimate.

\begin{deletedblock}
	To investigate the influence of hybrid prior modeling (Section~\ref{subsec:train_diff_model}) within our DiG framework, we compare three different setups.
	The first setup, denoted as {DiG (full)}, trains the denoisers in the diffusion models of $\sv_1$ and $\sv_2$ in a totally data-driven manner. The second setup, denoted as {DiG (sm)}, replaces the denoiser for $\sv_2$ by a totally model-driven MAP denoiser (cf. \eqref{eq:model-only-DM-training}) with a smoothness-promoting prior:
	\begin{equation}
		f_2\paren{\breve{\sv}_2} \coloneq \lambda\norm{ \Dm \breve{\sv}_2 }^2_2,
	\end{equation}
	where $\lambda>0$ is a tuning parameter and $\Dm$ is the 1d discrete difference operator. The third setup, denoted as {DiG (smld)}, adopts the same model-driven denoiser for $\sv_2$ as DiG (sm), and further trains the denoiser for $\sv_1$ with limited data using the regularized loss function \eqref{eq:data-plus-model-DM-training} equipped with a regularizer promoting sparsity in frequency domain (i.e., approximate periodicity):
	\begin{equation}
		f_1\paren{\breve{\sv}_1} \coloneq \norm{ \Fm \breve{\sv}_1 }_1,
	\end{equation}
	where $\Fm$ is the Fourier transform.
\end{deletedblock}

\deletedparagraph{3) Diffusion models} \added{The diffusion-model training settings used in this experiment are as follows.} The diffusion models for $\sv_1$ and $\sv_2$ follow standard implementations \cite{song2021}, except that the \deleted{architecture of denoisers are changed}\added{denoiser architecture is changed} from U-net to WaveNet \cite{van2016wavenet} so as to better handle 1d signals. The other settings are the same as the previous experiment. \deleted{In DiG~(full), MSDM, and DPnP, we use $50\text{K}$ samples to train the diffusion model for each component; in DiG~(sm), we use $50\text{K}$ samples only for training the diffusion model of $\sv_1$; and in DiG~(smld), we only use $5\text{K}$ samples for the diffusion model of $\sv_1$.}\added{All learned diffusion priors used in this experiment are trained with $50\text{K}$ samples per component.}

\paragraph*{3) \deleted{Results}\added{Main results}} We conduct experiments over a range of signal-to-interference ratio (SIR) and signal-to-noise ratio (SNR) conditions by multiplying the heartbeat signal, motion interference, and noise by proper positive constants. Here, SIR is the ratio between the average power of the heartbeat $\sv_1$ and that of the interference $\sv_2$, whilst SNR is the ratio between the power of $\sv_1$ and that of the noise $\vv$. For every SIR--SNR setting, we evaluate the estimation accuracy of the heartbeat $\sv_1$ for all methods by relative squared error (RSE, see \eqref{eq:error_signal}) over 200 problem instances.

Table~\ref{tab2} reports the RSE of $\sv_1$ for all algorithms under different SIR--SNR settings, and Fig.~\ref{fig1} illustrates \deleted{estimates of $\sv_1$ given by all algorithms in representative problem instances}\added{representative heartbeat recovery results under the most challenging setting, $(\mathrm{SIR},\mathrm{SNR})=(-40.1,-6.8)$ dB}. From Table~\ref{tab2} and Fig.~\ref{fig1}, the proposed \deleted{DiG algorithm consistently outperform all comparison methods across the full range of SIR and SNR settings, achieving more accurate heartbeat recovery especially under strong motion interference or low SNR.}\added{DiG algorithm achieves the lowest RSE in eight out of the nine tested SIR--SNR settings and is slightly below MSDM only in the remaining setting, which has the highest SIR and SNR, $(\mathrm{SIR},\mathrm{SNR})=(-20.1,13.2)$ dB. Overall, DiG provides more accurate heartbeat recovery than the competing methods, with particularly clear advantages under strong motion interference and low SNR.} \deleted{Moreover, DiG (sm) and DiG (smld) only shows moderate performance degradation compared to DiG (full), yet still outperforms other methods under low SIR--SNR settings while using much less training data, which demonstrates the usefulness of our hybrid prior modeling techniques.}

\added{Table~\ref{tab:heartbeat-runtime} reports the wall-clock inference time of all methods in this experiment. The four diffusion-based posterior samplers, MSDM, DPnP, DAPS, and DiG, have similar inference times under the same posterior-sample budget. EMD and VMD are substantially faster, reflecting the computational efficiency of these classical unsupervised decomposition methods, although their recovery accuracy is much lower in the low-SIR and low-SNR regimes in Table~\ref{tab2}. The GP baseline lies between these two groups in runtime.}

\begin{table}[t]
	\centering
	\caption{\added{Wall-clock inference time in the heartbeat extraction experiment. All entries are average seconds per test instance.}}
	\resizebox{\columnwidth}{!}{%
		\begin{tabular}{lccccccc}
			\toprule
			\textbf{Method}         & EMD   & VMD   & GP    & MSDM  & DPnP  & DAPS  & DiG   \\
			\midrule
			\textbf{Total/instance} & 0.006 & 0.034 & 1.798 & 2.756 & 2.702 & 2.672 & 2.708 \\
			\bottomrule
		\end{tabular}}
	\label{tab:heartbeat-runtime}
\end{table}

\begin{figure}[t]
	\centering
	\includegraphics[width=.9\linewidth]{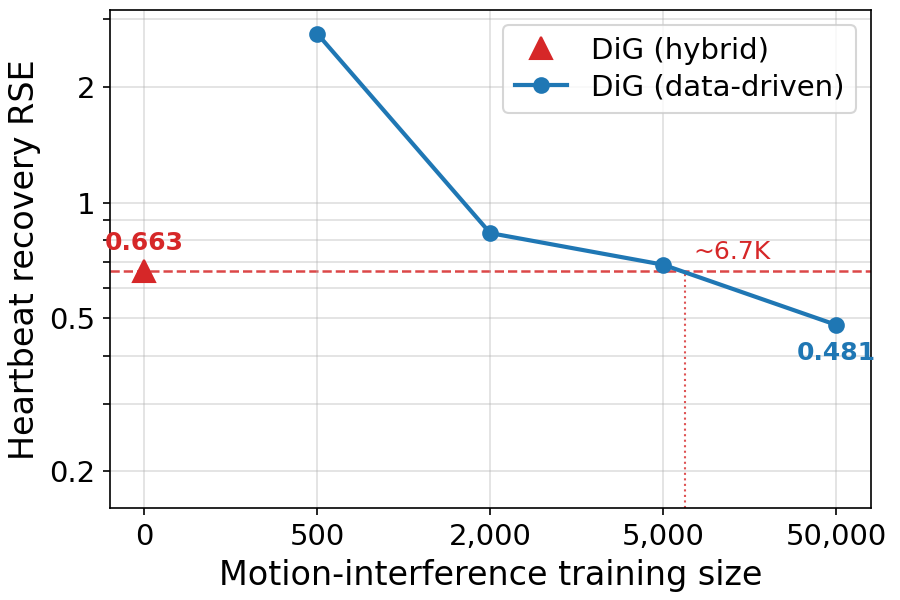}
	\caption{\added{Heartbeat recovery RSE versus motion-interference training size under $(\mathrm{SIR},\mathrm{SNR})=(-40.1,-6.8)$ dB. Blue: data-driven DiG with the indicated training size; red: DiG~(hybrid) with a model-driven MAP motion denoiser. The numeric labels on the red triangle and the rightmost blue marker report RSE values, whereas the \(\sim6.7\text{K}\) annotation reports the interpolated motion-training size at which data-driven DiG reaches the hybrid RSE.}}
	\label{fig:limited-data-heartbeat}
\end{figure}

\begin{addedblock}
	\paragraph*{4) Limited-data robustness and hybrid prior modeling} To examine how learned-prior errors affect the decomposition performance of DiG, we vary the motion-interference training size and evaluate the heartbeat recovery RSE, while keeping all other settings fixed as in the main SIR--SNR experiment. We use the most challenging SIR--SNR setting in Table~\ref{tab2}, namely $(\mathrm{SIR},\mathrm{SNR})=(-40.1,-6.8)$ dB. In addition, to assess the hybrid prior modeling strategy (Section~\ref{subsec:train_diff_model}), we replace the learned denoising network for the motion-interference component by the MAP denoiser in \eqref{eq:model-only-DM-training}. Specifically, we use the smoothness-promoting negative log-prior \(-\log f_2(\breve{\sv}_2)\coloneqq\lambda\|\Dm\breve{\sv}_2\|_2^2\), where \(\Dm\) is the one-dimensional discrete difference operator. The hyperparameter \(\lambda\) was selected empirically and then kept fixed across all test instances.\footnote{Systematic parameter-selection strategies for such regularization weights are a separate topic in regularization theory and inverse problems~\cite{hansen2001,zhang2010}.} We denote this implementation by DiG~(hybrid), and compare it with the data-driven DiG configuration.

	As shown in Fig.~\ref{fig:limited-data-heartbeat}, the heartbeat recovery error of data-driven DiG increases as the motion-interference training size decreases. The degradation is most pronounced at the smallest training size but remains progressive rather than abrupt, suggesting empirical robustness to learned-prior approximation errors in this tested setting. Although DiG~(hybrid) uses no motion-interference training samples and only a mismatched smoothness-promoting prior, it achieves a recovery error of \(0.663\), which is comparable to data-driven DiG trained with roughly \(6.7\text{K}\) motion-interference samples and lower than the errors of the competing diffusion-based samplers at the same SIR--SNR setting in Table~\ref{tab2}, showing that suitable model knowledge can effectively reduce the dependence on training data. Moreover, DiG~(hybrid) takes \(1.59\) seconds to process each instance, compared with \(2.71\) seconds for data-driven DiG reported in Table~\ref{tab:heartbeat-runtime}, indicating that replacing a learned denoising network by a model-based MAP denoiser can also reduce the inference cost.
\end{addedblock}

\begin{figure}[t]
	\centering
	\includegraphics[width=\linewidth]{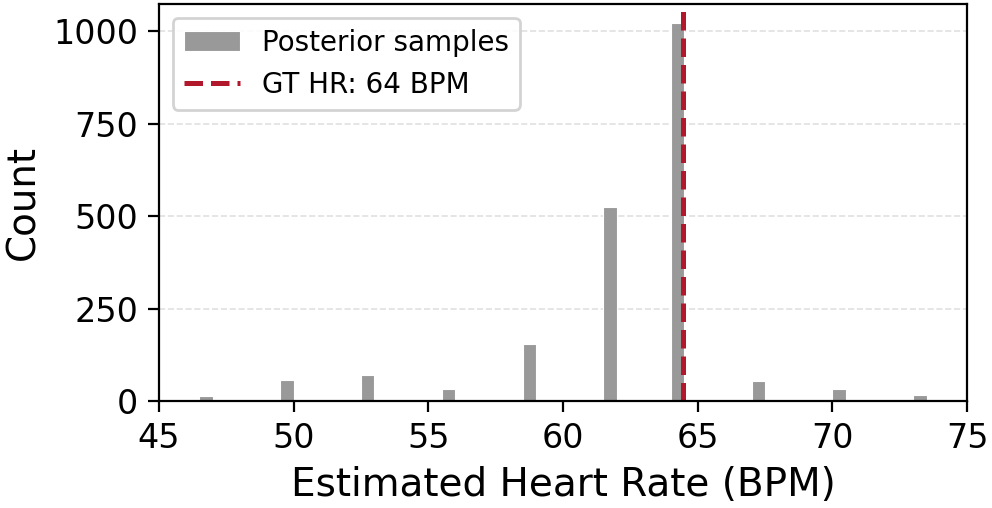}
	\caption{\added{Posterior histogram of the estimated heart rate for a representative instance under $(\mathrm{SIR},\mathrm{SNR})=(-40.1,-6.8)$ dB, computed from 2000 posterior heartbeat samples. Each posterior heartbeat sample gives one heart-rate estimate through the peak frequency of its spectrum, and the dashed line marks the ground-truth heart rate.}}
	\label{fig:heartbeat-uq-hr}
\end{figure}

\begin{addedblock}
	\paragraph*{5) Uncertainty quantification} To show how posterior samples can be used for uncertainty quantification, we take a representative observation from the most challenging $(\mathrm{SIR},\mathrm{SNR})=(-40.1,-6.8)$ dB setting in Table~\ref{tab2} and compute posterior uncertainty summaries from the posterior heartbeat samples generated by DiG. Considering the application goal of heartbeat extraction, we focus on the uncertainty of heart-rate estimation. For each posterior heartbeat sample, we compute its spectrum over the physiologically relevant frequency range and estimate the heart rate by the peak frequency. Fig.~\ref{fig:heartbeat-uq-hr} shows the posterior histogram of the resulting heart-rate estimates, together with the ground-truth heart rate. Such summaries can support downstream heartbeat decisions: concentrated posterior mass yields a more confident heart-rate estimate, whereas a diffuse or multimodal posterior can be used to flag uncertain segments for rejection, additional sensing, or human inspection. Beyond the task-specific histogram shown here, posterior samples can also support other UQ summaries, such as confidence intervals and empirical coverage diagnostics; see, e.g.,~\cite{zhang2024langevinized,dong2023}.
\end{addedblock}

\section{Conclusion}
\label{sec:conclusion}

We proposed a Bayesian framework for signal component decomposition, which combines Gibbs sampling with plug-and-play (PnP) diffusion priors. Our framework supports incorporating \added{component-wise }model-driven and data-driven \deleted{prior knowledge into the diffusion prior}\added{priors into diffusion models} in a unified manner. Moreover, the proposed diffusion-within-Gibbs (DiG) sampler allows component priors to be learned separately and flexibly combined \deleted{without retraining}\added{for different decomposition tasks at inference time}. Under suitable assumptions, we established the asymptotic consistency of the DiG sampler. Experiments have demonstrated the effectiveness of the proposed framework.

\appendices

\section{Proof of Lemma~\ref{lemma:cpdf-uk}}
\label{proof:cpdf-uk}

We rewrite $p_{\uv_k \mid \yv', \sv_{1:K}, \uv_{\neg k}}$ as
\begin{align}
	        & \;\; p_{\uv_k | \yv', \sv_{1:K}, \uv_{\neg k}} \cparen{ \cdot }{ \breve{\yv}, \breve{\sv}_{1:K}, \breve{\uv}_{\neg k} } \nonumber                                                         \\
	\propto & \;\; p_{\uv_k | \sv_{1:K}, \uv_{\neg k}} \cparen{ \cdot }{ \breve{\sv}_{1:K}, \breve{\uv}_{\neg k} } \nonumber                                                                            \\
	        & \hspace{2em}\times p_{\yv' | \uv_k, \sv_{1:K}, \uv_{\neg k}}\cparen{ \breve{\yv} }{ \cdot, \breve{\sv}_{1:K}, \breve{\uv}_{\neg k} } \nonumber                                            \\
	\propto & \;\; p_{\uv_k | \sv_k} \cparen{\cdot}{ \breve{\sv}_k } \times p_{\yv' | \uv_k, \sv_{1:K}, \uv_{\neg k}}\cparen{ \breve{\yv} }{ \cdot, \breve{\sv}_{1:K}, \breve{\uv}_{\neg k} } \nonumber \\
	\propto & \;\; \exp\paren{ -\inv{2\eta_k^2} \norm{ \paren{\cdot} - \breve{\sv}_k }^2_2 -\inv{2\sigma_v^2} \norm{\breve{\rv}'_{\neg k} - \Hm_k \paren{\cdot}  }^2_2 } \label{eq:cpdf-uk-temp1}
\end{align}
where the first proportionality follows from Bayes' rule, the second uses the independence among all components $\sv_k$ and $\uv_k\coloneq \sv_k + \vv_k$, the third follows immediately from \eqref{eq:observation-relax} and \(\vv_k\sim\Normal{\zerov}{\eta_k^2\Imat}\).

Next we expand the Gaussian density \(\NormalPdf{\cdot}{\muv_k}{\Sigmam_k}\) as
\begin{align}
	        & \;\; \NormalPdf{\cdot}{\muv_k}{\Sigmam_k} \nonumber                                                                                                                                               \\
	\propto & \;\; \exp\paren{ - \hf \paren{\cdot - \muv_k}^\top \Sigmam_k^{-1} \paren{ \cdot -\muv_k } } \nonumber                                                                                             \\
	\propto & \;\; \exp\paren{ -\hf \paren{\cdot}^\top \Sigmam_k^{-1} \paren{\cdot} + \muv_k^\top \Sigmam_k^{-1}\paren{\cdot} } \nonumber                                                                       \\
	\propto & \;\; \exp\paren{ -\hf \paren{\cdot}^\top \Sigmam_k^{-1} \paren{\cdot} + \paren{ \inv{\sigma_v^2} \Hm_k^\top \breve{\rv}'_{\neg k} + \inv{\eta_k^2} \breve{\sv}_k }^\top \paren{\cdot} } \nonumber \\
	\propto & \;\; \exp\paren{ -\inv{2\sigma_v^2} \paren{\cdot}^\top \Hm_k^\top\Hm_k \paren{\cdot} + \inv{\sigma_v^2} \paren{\breve{\rv}'_{\neg k}}^{\top} \Hm_k\paren{\cdot}  } \nonumber                      \\
	        & \hspace{2em} \times\exp\paren{ -\inv{2\eta_k^2} \paren{\cdot}^\top \paren{\cdot} + \inv{\eta_k^2} \breve{\sv}_k^\top \paren{\cdot} } \nonumber                                                    \\
	\propto & \;\; \exp\paren{ -\inv{2\eta_k^2} \norm{ \paren{\cdot} - \breve{\sv}_k }^2_2 -\inv{2\sigma_v^2} \norm{\breve{\rv}'_{\neg k} - \Hm_k \paren{\cdot}  }^2_2 } \label{eq:cpdf-uk-temp2}.
\end{align}
Comparing \eqref{eq:cpdf-uk-temp1} and \eqref{eq:cpdf-uk-temp2} completes the proof.

\section{Proof of Lemma~\ref{lemma:cpdf-sk}}
\label{proof:cpdf-sk}

We consider two cases.

\emph{Case 1: \(k\in\Hcal\).}
We have
\begin{align*}
	        & \;\; {p}_{\sv_{k} | \yv', \sv_{\neg k}, \uv_{\Hcal}}\cparen{\cdot}{\breve{\yv}, \breve{\sv}_{\neg k}, \breve{\uv}_{\Hcal}}                                             \\ \propto &\;\; p_{\sv_k | \sv_{\neg k}, \uv_{\Hcal}} \cparen{ \cdot }{ \breve{\sv}_{\neg k}, \breve{\uv}_{\Hcal} } \times p_{\yv' | \sv_k, \sv_{\neg k}, \uv_{\Hcal}} \cparen{ \breve{\yv} }{ \cdot, \breve{\sv}_{\neg k}, \breve{\uv}_{\Hcal} } \\
	\propto & \;\; p_{\sv_k | \uv_{k}} \cparen{ \cdot }{ \breve{\uv}_{k} } \times p_{\yv'| \sv_{\neg k}, \uv_{\Hcal}}\cparen{\breve{\yv}}{\breve{\sv}_{\neg k}, \breve{\uv}_{\Hcal}} \\
	\propto & \;\; p_{\sv_k | \sv_k + \eta_k \nv} \cparen{\cdot}{ \breve{\uv}_k },
\end{align*}
where the first proportionality follows from Bayes' rule, the second is due to the independence among all components $\sv_k$ and \eqref{eq:observation-relax}, the third uses \(\uv_k\coloneq \sv_k+\vv_k\) and \(\vv_k\sim\Normal{\zerov}{\eta_k^2\Imat}\).

\emph{Case 2: \(k\notin\Hcal\).}
We have
\begin{align*}
	        & \;\; {p}_{\sv_{k} | \yv', \sv_{\neg k}, \uv_{\Hcal}}\cparen{\cdot}{\breve{\yv}, \breve{\sv}_{\neg k}, \breve{\uv}_{\Hcal}}                                                                                                             \\
	\propto & \;\; p_{\sv_k | \sv_{\neg k}, \uv_{\Hcal}} \cparen{ \cdot }{ \breve{\sv}_{\neg k}, \breve{\uv}_{\Hcal} } \times p_{\yv' | \sv_k, \sv_{\neg k}, \uv_{\Hcal}} \cparen{ \breve{\yv} }{ \cdot, \breve{\sv}_{\neg k}, \breve{\uv}_{\Hcal} } \\
	\propto & \;\; p_{\sv_k}(\cdot) \times \exp\paren{- \inv{2\sigma_v^2} \norm{ \breve{\rv}'_{\neg k} - \paren{\cdot} }^2_2}                                                                                                                        \\
	\propto & \;\; p_{\sv_k | \sv_k + \sigma_v \nv} \cparen{\cdot}{ \breve{\rv}'_{\neg k} },
\end{align*}
where the first proportionality follows from Bayes' rule, the second is due to the component-independence assumption (Assumption~\ref{assump:independence}) and \eqref{eq:observation-relax}, and the last proportionality can be verified by expanding
\(p_{\sv_k \mid \sv_k+\sigma_v\nv}\cparen{\cdot}{\breve{\rv}'_{\neg k}}\) using Bayes' rule.

Combining the two cases completes the proof.

\section{Proof of Theorem~\ref{thm:consistency}}
\label{proof:consistency}

Let \deleted{us consider the $i$-th DiG iteration for $i\geq N_0$}\added{\(N_{\mathrm{w}}\) be the number of adaptive warm-up iterations in Section~\ref{subsec:warm-up}. For any \(i>N_{\mathrm{w}}\), the stationary sampling stage uses the fixed target parameters \(\sigma_v^{(i)}\equiv\sigma_v\) and \(\eta_k^{(i)}\equiv\eta_k\)}. Then according to the discussion in Section~\ref{sec:relaxation-technique}, the $i$-th DiG iteration \deleted{equipped with parameter annealing}\added{is} equivalent to one iteration of the Gibbs sampling method applied to $p_{\sv_{1:K},\uv_{\Hcal}|\yv'}\cparen{\cdot}{\breve{\yv}}$ (Algorithm~\ref{alg:Gibbs-relax}).

For \deleted{$i\geq N_0$}\added{\(i>N_{\mathrm{w}}\)}, let the Markov transition kernel from the $(i-1)$-th iterate $\paren{\breve{\sv}_{1:K},\breve{\uv}_{\Hcal}}$ of DiG, denoted by $\paren{\sv^{(i-1)}_{1:K}, \uv^{(i-1)}_{\Hcal}}$, to the $i$-th iterate, denoted by $\paren{\sv^{(i)}_{1:K},\uv^{(i)}_{\Hcal}}$, be $M$:
\begin{equation*}
	M(\zv,\zv')\coloneq p_{\sv^{(i)}_{1:K},\uv^{(i)}_{\Hcal} | \sv^{(i-1)}_{1:K}, \uv^{(i-1)}_{\Hcal} }\cparen{\zv'}{\zv}.
\end{equation*}
Since \deleted{each \(p_{\sv_k}\) is strictly positive everywhere}\added{each prior \(p_{\sv_k}\) has a strictly positive Lebesgue density on \(\R^{d_k}\)}, from Lemma~\ref{lemma:cpdf-uk} and Lemma~\ref{lemma:cpdf-sk}, for every $\breve{\yv}$, $\breve{\sv}_{1:K}$ and $\breve{\uv}_{\Hcal}$, the conditional distributions $p_{\uv_k|\yv',\sv_{1:K},\uv_{\neg k}}\cparen{\cdot}{\breve{\yv},\breve{\sv}_{1:K},\breve{\uv}_{\neg k}}$ and $p_{\sv_k|\yv',\sv_{\neg k},\uv_{\Hcal}}\cparen{\cdot}{\breve{\yv},\breve{\sv}_{\neg k}, \breve{\uv}_{\Hcal}}$ in Algorithm~\ref{alg:Gibbs-relax} are strictly positive everywhere. Thus for every $\zv$, $M(\zv,\cdot)$ is strictly positive everywhere, and one can verify that $M$ is $p_{\sv_{1:K},\uv_\Hcal|\yv'}\cparen{\cdot}{\breve{\yv}}$-irreducible and aperiodic.
Moreover, from the properties of the Gibbs sampler \cite[Sec. 2.2]{tierney1994}, the relaxed posterior $p_{\sv_{1:K},\uv_\Hcal|\yv'}\cparen{\cdot}{\breve{\yv}}$ is invariant for $M$. Then it follows directly from \cite[Thm. 1 and Cor. 1]{tierney1994} that
\begin{equation*}
	\lim_{N\to\infty}  \TV\paren{ p_{\sv^{(N)}_{1:K}, \uv^{(N)}_{\Hcal}}, p_{\sv_{1:K}, \uv_{\Hcal}|\yv'=\breve{\yv} } } = 0.
\end{equation*}
Since marginalization does not increase the total variation distance between two probability measures, we immediately obtain the convergence of $p_{\sv^{(N)}_{1:K}}$ to $p_{\sv_{1:K}|\yv'=\breve{\yv}}$.

\bibliographystyle{IEEEtran}
\bibliography{refs}

\end{document}